\documentclass[fleqn]{new_tlp}

\usepackage{algorithmicx}
\usepackage[ruled]{algorithm}
\usepackage{algpseudocode}
\usepackage{algpascal}
\usepackage{algc}
\usepackage{framed}
\usepackage{longtable}
\usepackage{microtype}
\usepackage{tikz}
\usepackage{amsmath,amsfonts,amssymb}
\usepackage{listings}
\usepackage{verbatim}
\usepackage{alltt}
\usepackage{stmaryrd}
\usepackage{mathtools}
\usepackage{dsfont}
\usepackage{thmtools, thm-restate}
\usepackage{graphicx}

\newtheorem{example}{Example}

\makeatletter
\def\infrule#1#2#3{\@ifnextchar[{\@infrule{#1}{#2}{#3}}{\@infrule{#1}{#2}{#3}[*]}}%

\def\@infrule#1#2#3[#4]{
 \par\bigbreak
 \vtop{
  \hangindent3em\hangafter2\leavevmode\null
  \textsf{#1:}\kern.5em\textbf{#2}\\[\smallskipamount]
  \ifx*#4 
   \null\qquad$#3$
  \else 
   \setbox30=\hbox{\qquad$#3$\qquad#4}%
   \ifdim\wd30>\hsize
    \null\qquad$#3$\par\kern-\parskip\smallskip#4
   \else
    \null\qquad$#3$\qquad#4
   \fi
  \fi
 }%
} 
\makeatother

\def\sol{\mathsf{solve}} 
\def\step{\mathsf{step}} 
\def\false{\mathsf{false}}  
\def\true{\mathsf{true}}  
\def\Bd{\mathsf{B}} 
\def\lit{L} 
 
\def\se{\tilde{E}} 
\def\substb{\theta} 

\def\Qr{\mathsf{G}} 
\def\Qrc{\hat{\mathsf{G}}}

\def\cVi{\mathcal{V}_\mathsf{T}}  
\def\cVs{\mathcal{V}_\mathsf{H}}  
\def\cVf{\mathcal{V}_\mathsf{F}}  
\def\cVv{\mathcal{V}}  
\def\var{\mathit{var}}  
\def\Fo{\mathcal{F}_{\sf{o}}} 
\def\Fu{\mathcal{F}_{\sf{u}}} 
\def\Fs{\mathcal{F}} 
\def\Ts{\mathcal{T}} 
\def\cC{\mathcal{C}}
\def\Ps{\mathcal{P}} 
\def\Fv{X} 
\def\fu{f_{\sf{u}}}  
\def\gu{g_{\sf{u}}}  
\def\hu{h_{\sf{u}}}  
\def\fo{f_{\sf{o}}}  
\def\go{g_{\sf{o}}}  
\def\ho{h_{\sf{o}}}  

\def\Gv{F}  
\def\Sv{\overline{x}} 
\def\Iv{x}  
\def\rh{\mathsf{R}} 
\def\concs{\cdot}  
\def\chs{+}  
\def\strs{*}  
\def\Input{\sf i} 
\def\Output{\sf o} 
\def\myrule{\leftarrow} 
\def\eql{\doteq}   
\def\eqls{\simeq}   
\def\inC{\mathrel{\mathsf{in}}}  
\def\pmc{\mathbf{c}} 
\def\conjc{\mathcal K}  
\def\prog{P} 
\def\rbr{\rangle} 
\def\lbr{\langle}  
\def\myequiv{=} 

\def\invar{\mathit{invar}}
\def\outvar{\mathit{outvar}}
\def\perm{\mathit{perm}}

\def\boxmy{\boxempty} 

\def\sep{\mathrel{\|}}
\def\defn{\mathit{defn}}
\def\mstruct{\mathfrak{S}}
\def\istruct{\mathfrak{I}}
\def\lang{\mathcal{L}}
\def\sdomi{D}
\def\sint{I}
\def\eps{\mathtt{eps}}

\def\ts{h}
\def\tseq{T}
\def\hdg{H}

\def\lra{\leftrightarrow}
\def\la{\leftarrow}

\def\alp{\mathcal{A}} 

\def\Fvy{Y}  
\def\Fvz{Z}  
\def\Svy{\overline{y}} 
\def\Svz{\overline{z}} 

\def\cB{\mathcal{B}}

\def\clphc{\rm{CLP}(\sf{H})}

\newcommand\regl[1]{[\![{#1}]\!]}

\newcommand\mset[1]{\{\!\vert{#1}\vert\!\}}
\def\cm{\mathit{cm}}
\def\size{\mathit{size}}

\def\constr{\mathcal C}
\def\constrd{\mathcal D}
\def\lf{\mathit{lf}}

\def\rpo{>_{\rm rpo}}

\mathchardef\mhyphen="2D

\newcommand\ASTART{\bigskip\noindent\begin{minipage}[b]{0.5\linewidth}}

\newcommand\AENDSKIP{\end{minipage}\bigskip}
\newcommand\AEND{\end{minipage}}

\submitted{28 September 2014}
\revised{07 January 2015}
\accepted{23 February 2015}

\begin{document}

\title{$\clphc${\rm :} Constraint Logic Programming for Hedges\footnotetext{This is an extended version of a paper presented at the Twelfth International Symposium on Functional and Logic Programming (FLOPS 2014), invited as a rapid publication in TPLP. The authors acknowledge the assistance of the conference chairs Michael Codish and Eijiro Sumii.}}

\author[B. Dundua, M. Florido, T. Kutsia, and M. Marin]
{BESIK DUNDUA\\
VIAM, Tbilisi State University, Georgia and
LIACC, University of Porto, Portugal\\
\email{bdundua@gmail.com}
\and
M\'{A}RIO FLORIDO\\
DCC-FC and LIACC,
University of Porto, Portugal
\email{amf@dcc.fc.up.pt}
\and
TEMUR KUTSIA\\
RISC,
Johannes Kepler University Linz,  Austria
\email{kutsia@risc.jku.at}
\and
MIRCEA MARIN\\
West University of Timi\c{s}oara, Romania
\email{mmarin@info.uvt.ro} 
}


\maketitle

\label{firstpage}

\vspace{-0.3cm}

\begin{abstract}
$\clphc$ is an instantiation of the general constraint logic programming scheme with the constraint domain of hedges. Hedges are finite sequences of unranked terms, built over variadic function symbols and three kinds of variables: for terms, for hedges, and for function symbols. Constraints involve equations between unranked terms and atoms for regular hedge language mem\-bership. We study algebraic semantics of $\clphc$ programs, define a sound, terminating, and incomplete constraint solver, investigate two fragments of constraints for which the solver returns a complete set of solutions, and describe classes of programs that generate such constraints. \vspace{0.2cm}

\noindent To appear in Theory and Practice of Logic Programming (TPLP).
\end{abstract}
\vspace{-0.35cm}

\section{Introduction}
\label{sect:intro}

Hedges are finite sequences of unranked terms. These are terms in which function symbols do not have a fixed arity: The same symbol may have a different number of arguments in different places. Manipulation of such expressions has been intensively studied in  recent years in the context of XML processing, rewriting, automated reasoning, knowledge representation, just to name a few.

When working with unranked terms, variables that can be instantiated with hedges (hedge variables) are a pragmatic necessity. In (pattern-based) programming, hedge variables help to write neat, compact code. Using them, for instance, one can extract duplicates from a list with just one line of a program. 
Several languages and formalisms operate on unranked terms and hedges. The programming language of Mathematica \cite{DBLP:books/daglib/0011324} is based on hedge pattern matching. Languages such as Tom \cite{DBLP:conf/rta/BallandBKMR07}, Maude \cite{DBLP:conf/maude/2007}, ASF+SDF \cite{DBLP:journals/entcs/BrandDHJJKKMOSVVV01} provide capabilities similar to hedge matching (via associative functions). $\rho$Log \cite{DBLP:journals/jancl/MarinK06} extends logic programming with hedge transformation rules, see also \cite{MarinKutsia03WIL}. XDuce \cite{DBLP:journals/jfp/HosoyaP03} enriches untyped hedge matching with regular expression types. The Constraint Logic Programming schema has been extended to work with hedges in CLP(Flex) \cite{DBLP:conf/coopis/CoelhoF04}, which is a basis for the XML processing language XCentric \cite{DBLP:conf/widm/CoelhoF07} and a Web site verification language VeriFLog \cite{DBLP:conf/apweb/CoelhoF06}.

The goal of this paper is to describe a precise semantics of constraint logic programs over hedges. We consider positive CLP programs with two kinds of primitive constraints: equations between hedges, and membership in a hedge regular language. Function symbols are unranked. Predicate symbols have a fixed arity. Terms may contain three kinds of variables: for terms (term variables), for hedges (hedge variables), and for function symbols (function variables). Moreover, we may have function symbols whose argument order does not matter (unordered symbols): a kind of generalization of the commutativity property to unranked terms. As it turns out, such a language is very flexible and permits to write short, yet quite clear and intuitive code: One can see examples in Sect.~\ref{sect:examples}. 
We call this language $\clphc$, for CLP over hedges. It generalizes CLP(Flex) with function variables, unordered functions, and membership constraints. Hence, as a special case, our paper describes the semantics of CLP(Flex). Moreover, as hedges generalize strings, $\clphc$ can be seen also as a generalization of CLP over strings $\rm{CLP}(\mathcal{S})$~\cite{DBLP:conf/slp/Rajasekar94}, string processing features of Prolog III~\cite{DBLP:journals/cacm/Colmerauer90}, and CLP over regular sets of strings CLP($\rm\Sigma^*$)~\cite{DBLP:conf/iclp/Walinsky89}.

Note that some of these languages allow an explicit size factor for string variables, restricting the length of strings they can be instantiated with. We do not have size factors, but can express this information easily with constraints. For instance, to indicate the fact that a hedge variable $\Sv$ can be instantiated with a hedge of  minimal length 1 and  maximal length 3, we can write a disjunction $\Sv \eql x \lor \Sv \eql (x_1,x_2)\lor \Sv \eql (x_1,x_2,x_3) $, where the lower case $x$'s are term variables.

Flexibility and the expressive power of $\clphc$ has its price: Equational constraints with hedge variables, in general, may have infinitely many solutions~(\citeNP{DBLP:conf/aisc/Kutsia04}; \citeyearNP{DBLP:journals/jsc/Kutsia07}). Therefore, any complete equational constraint solving procedure with hedge variables is nonterminating. The solver we describe in this paper is sound and terminating, hence incomplete for arbitrary constraints. However, there are fragments of constraints for which it is complete, i.e., computes all  solutions. One such fragment is so called well-moded fragment, where variables in one side of equations (or in the left hand side of the membership atom) are guaranteed to be instantiated with ground expressions at some point. This effectively reduces constraint solving to hedge matching~(\citeNP{KutsiaMarin05UNIF}; \citeyearNP{DBLP:conf/lpar/KutsiaM05}), plus some early failure detection rules. Another fragment for which the solver is complete is named after the Knowledge Interchange Format, KIF~\cite{genesereth92knowledge}, where hedge variables are permitted only in the last argument positions. We identify forms of $\clphc$ programs which give rise to well-moded or KIF constraints.\footnote{Conceptually, such an approach can be seen to be similar to, e.g., Miller's approach to higher-order logic programming~\cite{DBLP:journals/logcom/Miller91}, where the fragment $L_\lambda$ uses unitary unification for higher-order patterns instead of undecidable higher-order unification.}

We can easily model lists with ordered function symbols and multisets with the help of unordered ones. In fact, since we may have several such symbols, we can directly model colored multisets. Constraint solving over lists, sets, and multisets 
has been intensively studied, see, e.g.,~\cite{DBLP:journals/tocl/DovierPR08} and references there, and the CLP schema can be extended to accommodate them. In our case, an advantage of using hedge variables in such terms is that hedge variables can give immediate access to collections of subterms via unification. It is very handy in programming.

This paper is an extended and revised version of~\cite{DBLP:conf/flops/DunduaFKM14}. It is organized as follows: After establishing the terminology in Section~\ref{sect:preliminaries}, we give two motivating examples in Section~\ref{sect:examples} to illustrate $\clphc$. The algebraic semantics is studied in Section~\ref{sect:semantics}. The constraint solver is introduced in Section~\ref{sect:solver}. The operational semantics of $\clphc$ is described in Section~\ref{sect:operational}. In Sections~\ref{sect:well-moded} and \ref{sect:kif}, we introduce the well-moded and KIF fragments, respectively. Section~\ref{sect:conclusion} contains concluding remarks.

\section{Preliminaries}
\label{sect:preliminaries}

For common notation and definitions, we mostly follow~\cite{DBLP:journals/jlp/JaffarMMS98}.
The alphabet $\alp$ consists of the following pairwise disjoint sets of symbols:
\begin{itemize}
 \item $\cVi$: term variables, denoted by $x,y,z,\ldots$,
 \item $\cVs$: hedge variables, denoted by $\Sv,\Svy,\Svz,\ldots$,
 \item $\cVf$: function variables, denoted by $\Fv,\Fvy,\Fvz,\ldots$,
 \item $\Fu$: unranked unordered function symbols, denoted by $\fu,\gu,\hu,\ldots$, 
 \item $\Fo$: unranked ordered function symbols, denoted by $\fo,\go,\ho,\ldots$, 
 \item $\Ps$: ranked predicate symbols, denoted by $p,q,\ldots$.
\end{itemize}
The sets of variables are countable, while the sets of function and predicate symbols are finite.
In addition, $\alp$ also contains 
 \begin{itemize}
 \item The propositional constants $\true$ and $\false$, the binary equality predicate $\eql$, and the unranked membership predicate $\inC$.
 \item Regular operators: $\eps, \concs, \chs, \strs$.
 \item Logical connectives and quantifiers: $\neg$, $\lor$, $\land$, $\to$, $\lra$, $\exists$, $\forall$.
 \item Auxiliary symbols: parentheses and the comma.
\end{itemize}
\emph{Function symbols,} denoted by $f,g,h,\ldots$, are elements of the set $\Fs=\Fu\cup\Fo.$ A {\em variable} is an element of the set $\cVv=\cVi \cup \cVs \cup \cVf$. A {\em functor}, denoted by $F$, is a common name for a function symbol or a function variable. 

We define \emph{terms, hedges,} and other syntactic categories over $\alp$ as follows:
\begin{alignat*}{3}
 t & ::= \Iv \mid f(\hdg)  \mid \Fv(\hdg)  & &\text{Term}\\
 \tseq & ::= t_1,\ldots,t_n \quad (n\ge 0) & & \text{Term sequence}\\
 \ts & ::= t \mid \Sv  & & \text{Hedge element}\\
 \hdg & ::= \ts_1,\ldots,\ts_n \quad (n\ge 0) &  \qquad & \text{Hedge}
\end{alignat*}
We denote the set of terms by $\Ts(\Fs,\cVv)$ and the set of ground (i.e., variable-free) terms by $\Ts(\Fs)$. Besides the letter $t$, we use also $r$ and $s$ to denote terms. 

We make a couple of conventions to improve readability. The empty hedge is written as $\epsilon$. The terms of the form $a(\epsilon)$ and $X(\epsilon)$ are abbreviated as $a$ and $X$, respectively. We put parentheses around hedges, writing, e.g., $(f(a),\Sv,b)$ instead of $f(a),\Sv,b$. For hedges $H=(h_1,\ldots,h_n)$ and $H'=(h'_1,\ldots,h'_{n'})$, the notation $(H,H')$ stands for the hedge $(h_1,\ldots,h_n,h'_1,\ldots,h'_{n'})$.

Two hedges are \emph{disjoint} if they do not share a common element. For instance, $(f(a),x,b)$ and $(f(x),f(b,f(a)))$ are disjoint, whereas $(f(a),x,b)$ and $(f(b),f(a))$ are not, because $f(a)$ is their common element.

An \emph{atom} is a formula of the form $p(t_1,\ldots,t_n)$, where $p\in\Ps$ is an $n$-ary predicate symbol. Atoms are denoted by $A$. 

\emph{Regular hedge expressions} $\rh$ are defined inductively:
\begin{alignat*}{1}
     \rh ::={} & \eps \mid (\rh \concs \rh) \mid \rh \chs \rh \mid \rh^\strs \mid f(\rh) 
\end{alignat*}
where the dot $\concs$  stands for concatenation,  $\chs$ for choice, and $\strs$ for repetition.
\emph{Primitive constraints} are either term equalities $\doteq(t_1,t_2)$ or membership for hedges $\mathsf{in}(\hdg, \rh)$. They are written in infix notation, such as $t_1\eql t_2$, and $\hdg \inC \rh$.  

A \emph{literal} $L$ is an atom or a primitive constraint. \emph{Formulas} are defined as usual.  A \emph{constraint} is an arbitrary first-order formula built over $\true$, $\false$, and primitive constraints. 

The set of free variables of a syntactic object $O$ is denoted by $\var(O)$. We let $\exists_V N$ denote the formula $\exists v_1\cdots \exists v_n N$, where $V=\{v_1,\ldots,v_n\}\subset \cVv$. $\overline{\exists}_V N$ denotes $\exists_{\var(N)\setminus V} N$. We write $\exists N$ (resp. $\forall N$) for the existential (resp. universal) closure of $N$.
We refer to a language over the alphabet $\alp$ as $\lang(\alp)$.

A \emph{substitution} is a mapping from term variables to terms, from hedge variables to hedges, and from function variables to functors,  such that all but finitely many variables are mapped to themselves. We use lower case Greek letter to denote them. 

For an expression (i.e., a term, hedge, functor, literal, or a formula) $e$ and a substitution $\sigma$, we write $e\sigma$ for the \emph{instance} of $e$ under $\sigma$. This is a standard operation that replaces in $e$ each free occurrence of a variable $v$ by its image under $\sigma$, i.e., by $\sigma(v)$. If needed, bound variables are renamed to avoid variable capture. For instance, for the constraint $\constr = \forall \Iv. f(\Fv(a,\Sv), \Sv) \eql f(g(\Svy,a,b,\Iv), b,\Iv)$ and the substitution $\sigma=\{\Fv \mapsto g,\allowbreak  \Sv \mapsto (b,\Iv), \Svy \mapsto \epsilon, \Iv \mapsto f(c)\}$, we have $\constr\sigma=\forall z. f(g(a,b,\Iv), b,\Iv) \eql f(g(a,b,z), b,z)$. A substitution $\sigma$ is \emph{grounding} for an expression $e$ if $e\sigma$ is a ground expression. 

A \emph{(constraint logic) program} is a finite set of \emph{rules} of the form $\forall(L_1\land \cdots \land L_n \to A)$, $n\geq 0$, usually written as $A \myrule L_1,\ldots, L_n$,  where $A$ is an atom and $L_1, \ldots,L_n$ are literals other than $\true$ and $\false$. A \emph{goal} is a formula of the form $\exists(L_1 \land\cdots \land L_n)$, $n\ge 0$, usually written as $L_1,\ldots,L_n$ where $L_1, \ldots,L_n$ are literals other than $\true$ and $\false$.

We say a variable is \emph{solved} in a conjunction of primitive constraints $\conjc=\pmc_1\land\cdots \land \pmc_n$, if there is a 
$\pmc_i,$ $ 1\le i \le n$,  such that
\begin{itemize}
\item the variable is $x$, $\pmc_i=\Iv\eql t$, and $\Iv$ occurs neither in $t$ nor elsewhere in $\conjc $, or
\item the variable is $\Sv$, $\pmc_i=\Sv\eql \hdg$, and $\Sv$ occurs neither in $\hdg$ nor elsewhere in $\conjc$, or
\item the variable is $\Fv$, $\pmc_i=\Fv \eql \Gv $ and $\Fv $ occurs neither in $\Gv$ nor elsewhere in $\conjc$, or
\item the variable is $\Iv$, $\pmc_i=\Iv \inC f(\rh)$ and $\Iv $ does not occur in membership constraints elsewhere in $\conjc $, or
\item the variable is $\Sv$, $\pmc_i=\Sv \inC \rh$, $\Sv$ does not occur in membership constraints elsewhere in $\conjc $, and $\rh$ has the form $\rh_1 \concs \rh_2$ or $\rh_1^*$.
\end{itemize}
In this case we also say that $\pmc_i$ is \emph{solved in} $\conjc$. Moreover, $\conjc$ is called \emph{solved} if  for any $1\le i \le n$, $\pmc_i$ is solved in it. $\conjc$ is \emph{partially solved}, if  for any $1\le i \le n$,  $\pmc_i$ is solved in $\conjc$, or has one of the following forms:
\allowdisplaybreaks
\begin{itemize}
\item Membership atom:
  \begin{itemize}
   \item $\fu(\hdg_1,\Sv,\hdg_2) \inC \fu(\rh)$.
    \item $(\Sv,\hdg) \inC \rh $  where  $\hdg \neq \epsilon $ and $\rh$ has the form  $\rh_1 \concs \rh_2 $ or $\rh_1^\strs$.
\end{itemize}
\item Equation: 
   \begin{itemize}
     \item $ (\Sv,\hdg_1)\eql (\Svy,\hdg_2) $ where $\Sv \not = \Svy$, $\hdg_1\neq \epsilon$ and $\hdg_2\neq \epsilon$.
     \item $ (\Sv,\hdg_1)\eql (\tseq,\Svy,\hdg_2) $, where $\Sv\not \in \var(\tseq)$, $\hdg_1\neq \epsilon$, and $\tseq\neq \epsilon$.  The variables $\Sv $ and $\Svy$ are not necessarily distinct.
  \item $\fu(\hdg_1,\Sv,\hdg_2)\eql \fu(\hdg_3,\Svy,\hdg_4)$ where $(\hdg_1,\Sv,\hdg_2)$ and $(\hdg_3,\Svy,\hdg_4)$ are disjoint.
\end{itemize}
 \end{itemize}

A constraint is \emph{solved}, if it is either $\true$ or a non-empty quantifier-free disjunction of solved conjunctions. A constraint is \emph{partially solved}, if it is either $\true$ or a non-empty quantifier-free disjunction of partially solved conjunctions.

\section{Motivating Examples}
\label{sect:examples}

In this section we illustrate the expressive power of $\clphc$ by two examples: the rewriting of terms from some regular hedge language and an implementation of the recursive path ordering with status.

\begin{example}
  \label{exmp:rewriting}
  The general rewriting mechanism can be implemented with two $\clphc$ clauses: The base case \[\mathit rewrite(x,y) \la rule(x,y)\] and the recursive case \[\mathit rewrite(\Fv(\Sv,x,\Svy), \Fv(\Sv,y,\Svy)) \la rewrite(x, y),\] where $x,y$ are term variables, $\Sv,\Svy$ are hedge variables, and $\Fv$ is a function variable. It is assumed that there are clauses which define the \emph{rule} predicate. The base case says that a term $x$ can be rewritten to $y$ if there is a rule which does it. The recursive case rewrites a nondeterministically selected subterm $x$ of the input term to $y$, leaving the context around it unchanged. Applying the base case before the recursive case gives the outermost strategy of rewriting, while the other way around implements the innermost one.

An example of the definition of the \emph{rule} predicate is \[\mathit rule(\Fv(\Sv_1,\Sv_2),\Fv(\Svy)) \la\ \Sv_1 \inC f(a^\strs)\concs b^\strs,\ \Sv_1 \eql (x,\Svz),\ \Svy \eql (x,f(\Svz)),\]
where the constraint\footnote{In the notation defined in the previous section, strictly speaking, we need to write this constraint as $f(a(\eps)^\strs)\concs b(\eps)^\strs$. However, for brevity and clarity of the presentation we omit $\eps$ here.} $\Sv_1 \inC f(a^\strs)\concs b^\strs$ requires $\Sv_1$ to be instantiated by hedges from the language generated by the regular hedge expression $f(a^\strs)\concs b^\strs$ (that is, from the language $\{f, f(a), f(a,a), \ldots, (f,b), (f(a),b), \ldots, (f(a,\ldots,a),b,\ldots,b), \ldots\}$).

With this program, the goal $\mathit \la\ rewrite(f(f(f(a,a),b)), x)$ has two answer substitutions: $\{x\mapsto f(f(f(a,a),f))\}$ and $\{x\mapsto f(f(f(a,a),f(b)))\}$. To obtain them, the goal is first transformed by the recursive clause, leading to the new goal $\mathit \la\ rewrite(f(f(a,a),b), y)$ together with the constraint $x\eql f(y)$ for $x$. The next transformation is performed by the base case of the \emph{rewrite} predicate, resulting into the goal $\mathit \la\ rule(f(f(a,a),b), y)$. This goal is then transformed by the \emph{rule} clause, which gives the constraint $\Fv(\Sv_1,\Sv_2) \eql f(f(a,a),b) \land y \doteq \Fv(\Svy) \land \Sv_1 \inC f(a^\strs)\concs b^\strs \land \Sv_1 \eql (x',\Svz)\land  \Svy \eql (x',f(\Svz)) \land x\eql f(y)$. This constraint has two solutions, depending whether $\Sv_1$ equals $f(a,a)$ or to $(f(a,a),b)$. From one we get $x\eql f(f(f(a,a),f))$, and from the other $x\eql f(f(f(a,a),f(b)))$. These solutions give the above mentioned answers.
\end{example}

\begin{example}
  \label{exmp:rpo}
  The recursive path ordering (rpo) $\rpo$ is a well-known term ordering~\cite{DBLP:journals/tcs/Dershowitz82} used to prove termination of rewriting systems. Its definition is based on a precedence order $\succ$ on function symbols, and on extensions of $\rpo$ from terms to tuples of terms. There are two kinds of extensions: lexicographic $\rpo^{\mathit{lex}}$, when terms in tuples are compared from left to right, and multiset $\rpo^{\mathit{mul}}$, when terms in tuples are compared disregarding the order. The status function $\tau$ assigns to each function symbol either \emph{lex} or \emph{mul} status. Then for all (ranked) terms $s,t$, we define $s \rpo t$,  if $s=f(s_1,\ldots, s_m)$ and
  \begin{enumerate}
    \item either $s_i=t$ or $s_i \rpo t$ for some $s_i$, $1\le i \le m$, or
    \item $t=g(t_1,\ldots,t_n)$, $s\rpo t_i$ for all $i, 1\le i \le n$, and either
      \begin{enumerate}
        \item $f \succ g$, or (b) $f=g$ and $(s_1,\ldots, s_n) \rpo^{\tau(f)} (t_1,\ldots,t_n)$.
      \end{enumerate}
  \end{enumerate}

To implement this definition in $\clphc$, we use the predicate \emph{rpo} for $\rpo$ between two terms, and four helper predicates: $\mathit{rpo\_all}$ to implement the comparison $s\rpo t_i$ for all $i$; $\mathit{prec}$ to implement the comparison depending on the precedence; $\mathit{ext}$ to implement the comparison with respect to an extension of $\rpo$; and \emph{status} to give the status of a function symbol. The predicate \emph{lex} implements $\rpo^{\mathit{lex}}$ and \emph{mul} implements $\rpo^{\mathit{mul}}$. The symbol $\langle \rangle$ is an unranked function symbol, and $\{\}$ is an unordered unranked function symbol. As one can see, the implementation is rather straightforward and closely follows the definition. $\rpo$ requires four clauses, since there are four alternatives in the definition:
\begin{alignat*}{1}
1. \quad & \mathit{rpo(\Fv(\Sv,x,\Svy), x).} \\
         & \mathit{rpo(\Fv(\Sv,x,\Svy), y) \la rpo(x, y).}\\
2a. \quad & \mathit{rpo(\Fv(\Sv), \Fvy(\Svy)) \la rpo\_all(\Fv(\Sv), \langle \Svy \rangle), prec(\Fv,\Fvy).}\\
2b. \quad & \mathit{rpo(\Fv(\Sv), \Fv(\Svy)) \la rpo\_all(\Fv(\Sv), \langle \Svy \rangle), ext(\Fv(\Sv), \Fv(\Svy)).}
\end{alignat*}

$\mathit{rpo\_all}$ is implemented with recursion:
\begin{alignat*}{1}
& \mathit{rpo\_all(x, \langle\, \rangle).} \\
& \mathit{rpo\_all(x, \langle y,\Svy \rangle) \la rpo(x,y), rpo\_all(x, \langle \Svy \rangle).}
\end{alignat*}
The definition of \emph{prec} as an ordering on finitely many function symbols is straightforward. More interesting is the definition of \emph{ext}:
\begin{alignat*}{1}
& \mathit{ext(\Fv(\Sv), \Fv(\Svy)) \la status(\Fv,lex), lex(\langle\Sv \rangle, \langle\Svy \rangle).} \\
& \mathit{ext(\Fv(\Sv), \Fv(\Svy)) \la status(\Fv,mul), mul(\{\Sv \}, \{\Svy \}).}
\end{alignat*}
\emph{status} can be given as a set of facts, \emph{lex} needs one clause, and \emph{mul} requires three: 
\begin{alignat*}{1}
& \mathit{lex(\langle\Sv,x,\Svy \rangle, \langle\Sv,y,\Svz \rangle) \la rpo(x,y).} \\
& \mathit{mul(\{x,\Sv\},\{\}).} \\
& \mathit{mul(\{x,\Sv\},\{x,\Svy\}) \la mul(\{\Sv\},\{\Svy\}).}\\
& \mathit{mul(\{x,\Sv\},\{y,\Svy\}) \la rpo(x,y),\, mul(\{x,\Sv\},\{\Svy\}).}
\end{alignat*}
\end{example}
That's all. This example illustrates the benefits of all three kinds of variables we have and unordered function symbols.

\section{Algebraic Semantics}
\label{sect:semantics}

For a given set $S$, we denote by $S^*$ the set of finite, possibly empty, sequences of  elements of $S$, and by $S^n$ the set of sequences of length $n$ of elements of $S$. The empty sequence of symbols from any set $S$ is denoted by $\epsilon$. Given a sequence $s=(s_1,s_2,\ldots,s_n)\in S^n$, we denote by $\perm(s)$ the set of sequences $\{(s_{\pi(1)},s_{\pi(2)},\ldots,s_{\pi(n)})\mid \pi$ is a permutation of $\{1,2,\ldots,n\}\}.$

A \emph{structure} $\mstruct$ for a language $\lang(\alp)$ is a tuple $\langle \sdomi,\sint \rangle$ made of  a non-empty carrier set of \emph{individuals} and an interpretation function $\sint$ that maps each function symbol $f\in \Fs$ to a function $\sint(f):\sdomi^*\to\sdomi$, and  each $n$-ary predicate symbol $p\in \Ps$ to an $n$-ary relation $\sint(p)\subseteq \sdomi^n$. Moreover, if $f\in\Fu$ then $\sint(f)(s)=\sint(f)(s')$ for all $s\in D^*$ and $s'\in\perm(s).$
A \emph{variable assignment} for such a structure is a  function with domain $\cVv$ that maps  term variables to elements of $\sdomi$, hedge variable to elements of $\sdomi^*$, and  function variables  to functions from $\sdomi^*$ to $\sdomi$.

The interpretations of our syntactic categories w.r.t. a structure $\mstruct=\langle \sdomi,\sint \rangle$ and variable assignment $\sigma$ is shown below. The interpretations $\regl{\hdg}_{\mstruct,\sigma}$ of hedges (including terms)  is defined as follows:
\begin{alignat*}{7}
& \regl{v}_{\mstruct,\sigma}  :=\sigma(v), \text{ where } v \in \cVi \cup \cVs.\\
& \regl{f(\hdg)}_{\mstruct,\sigma}   := \sint(f)(\regl{\hdg}_{\mstruct,\sigma}).\\
& \regl{\Fv(\hdg)}_{\mstruct,\sigma}    := \sigma(\Fv)(\regl{\hdg}_{\mstruct,\sigma}).\\
&  \regl{(\ts_1,\ldots, \ts_n)}_{\mstruct,\sigma} := (\regl{\ts_1}_{\mstruct,\sigma},\ldots,\regl{\ts_n}_{\mstruct,\sigma}).
\end{alignat*}

Note that terms are interpreted as elements of $\sdomi$ and hedges as elements of $\sdomi^*$. We may omit $\sigma$ and write simply $\regl{E}_\mstruct$ for the interpretation of a ground expression~$E$. The interpretation of regular expressions is defined as follows:
\allowdisplaybreaks
\begin{alignat*}{7}
  & \regl{\mathtt{eps}}_\mstruct  := \{\mathtt{\epsilon}\}.\\
  & \regl{f(\rh)}_\mstruct  := \{\sint(f)(\hdg) \mid \hdg\in \regl{\rh}_\mstruct\}. \\
  & \regl{\rh_1 \chs \rh_2}_\mstruct  := \regl{\rh_1}_\mstruct \cup \regl{\rh_2}_\mstruct.  \\
  & \regl{\rh_1 \concs \rh_2}_\mstruct  := \{(\hdg_1,\hdg_2) \mid \hdg_1\in \regl{\rh_1}_\mstruct, \hdg_2\in\regl{\rh_2}_\mstruct\}.\\
  & \regl{\rh^\strs}_\mstruct := \regl{\rh}^*_\mstruct.
\end{alignat*}

Primitive constraints are interpreted with respect to a structure $\mstruct$ and variable assignment $\sigma$ as follows: 
\begin{alignat*}{1}
 & \mstruct \models_\sigma t_1\eql t_2 \text{ iff } \regl{t_1}_{\mstruct,\sigma}=\regl{t_2}_{\mstruct,\sigma}.\\
 & \mstruct \models_\sigma \hdg\inC \rh \text{ iff } \regl{\hdg}_{\mstruct,\sigma}\in\regl{\rh}_\mstruct.\\
 & \mstruct \models_\sigma p(t_1,\ldots,t_n) \text{ iff } \sint(p)(\regl{t_1}_{\mstruct,\sigma},\ldots,\regl{t_n}_{\mstruct,\sigma}).
\end{alignat*}

The notions  $\mstruct \models N$ for {validity of an arbitrary formula $N$ in $\mstruct$,} and $\models N$ for validity of $N$ in any structure are defined in the standard way. 

An \emph{intended structure} is a structure $\istruct$ with the carrier set $\Ts(\Fs)$ and interpretations  $\sint$  defined for every $f\in \Fs$ by $\sint(f)(\hdg):=f(\hdg)$.   Thus, intended structures identify terms and hedges  by themselves. Also, if $\rh$ is any regular hedge expression then $\regl{\rh}_{\istruct}$ is the same in all intended structures, and will be denoted  by $\regl{\rh}$. Other remarkable properties of intended structures $\istruct$ are: Variable assignments are substitutions, $\istruct\models_\vartheta t_1\doteq t_2$ iff $t_1\vartheta=t_2\vartheta$, and $\istruct\models_\vartheta \hdg\mathop{\mathsf{in}} \rh$ iff $\hdg\vartheta\in\regl{\rh}$. 

Given a program $P$, its Herbrand base $\cB_P$ is, naturally, the set of all atoms $p(t_1,\ldots,t_n)$, where $p$ is an $n$-ary user-defined predicate in $P$ and $(t_1,\ldots,t_n)\in \Ts(\Fs)^n$. Then an intended interpretation of $P$ corresponds uniquely to a subset of $\cB_P$. An \emph{intended model} of $P$ is an intended interpretation of $P$ that is its model. 

As usual, we will write $P\models G$ if $G$ is a goal which holds in every model of $P$. Since our programs consist of positive clauses, the following facts hold:
\begin{enumerate}
\item Every program $P$ has a least intended model, which we denote by $lm(P)$.
\item If $G$ is a goal then $P\models G$ iff $lm(P)$ is a model of $G$.
\end{enumerate}

A ground substitution $\vartheta$ is an \emph{intended solution} (or simply \emph{solution}) of a constraint $\cC$ if $\istruct\models \constr\vartheta$ for all intended structures $\istruct$. 
\begin{restatable}{theorem}{oneThmSatisfiable}
\label{thm:sat}
If the constraint $\constr$ is solved, then $\istruct\models \exists \constr$ holds for all intended structures $\istruct$.
\end{restatable}

\section{Solver}
\label{sect:solver}

In this section we present a constraint solver for quantifier-free constraints in DNF. It is based on rules, transforming a constraint in \emph{disjunctive normal form} (DNF) into a constraint in DNF. We say a constraint is in DNF, if it has a form $\conjc_1 \lor \cdots \lor  \conjc_n$, where $\conjc$'s are conjunctions of $\true$, $\false$, and primitive constraints. The number of rules is not small (as it is usual for such kind of solvers, cf., e.g., \cite{DBLP:journals/toplas/DovierPPR00,DBLP:journals/jsc/Comon98a}). To make their comprehension easier, we group them so that similar ones are collected together in subsections. Within each subsection, for better readability, the rule groups are put between horizontal lines.

Before going into the details, we introduce a more conventional way of writing expressions, some kind of syntactic sugar, that should make reading easier. Instead of $F_1() \eql F_2()$ and $\fo(\hdg_1) \eql \fo(\hdg_2)$ we write $F_1 \eql F_2$ and $\hdg_1 \eql \hdg_2$ respectively. The symmetric closure of the relation $\eql$ is denoted by $\eqls$. The rules are applied in any context, i.e., they behave as rewrite rules. Moreover, when a rule applies to a conjunction of the form $L \land \conjc$, it is intended to act on an entire conjunct of the DNF, modulo associativity and commutativity of $\land$. These assumptions guarantee that the constraint obtained after each rule application is again in DNF.

\subsection{Rules}
\label{subsect:rules}
\subsubsection*{Logical Rules.} There are eight logical rules which are applied at any depth in constraints, modulo associativity and commutativity of disjunction and conjunction. $N$ stands for any formula. We denote the whole set of rules by {\sf Log}. %

\vskip\medskipamount
\leaders\vrule width \textwidth \vskip0.4pt 
\leaders\vrule width 0.1\textwidth \vskip4pt 
\nointerlineskip
\begin{alignat*}{15}
& N \land N \leadsto N           & \qquad \qquad  & N \lor N \leadsto N  & \\
& \false \land N \leadsto \false &                & \false \lor N \leadsto N \\
& \true \land N \leadsto N       &                & \true \lor N \leadsto \true \\
& \hdg \eql \hdg \leadsto \true   &                & \epsilon \inC \rh  \leadsto   \true, \text{ if } \epsilon\in\regl{\rh}
\vspace{-0.5cm}
\end{alignat*}
\noindent\rule{0.9\textwidth}{0.4pt}\rule{0.1\textwidth}{4pt}
\subsubsection*{Failure Rules.} 
The first two rules perform occurrence check, rules {\sf (F3)} and {\sf (F5)} detect function symbol clash, and rules {\sf (F4)}, {\sf (F6)}, {\sf (F7)}  detect inconsistent primitive constraints. We denote the set of rules {\sf (F1)}--{\sf (F7)} by {\sf Fail}.
\vskip\medskipamount
\leaders\vrule width \textwidth \vskip0.4pt 
\leaders\vrule width 0.1\textwidth \vskip4pt 
\nointerlineskip
\begin{alignat*}{7}
{\sf (F1)} &\quad & & \Iv \eqls (\hdg_1,F(\hdg),\hdg_2)  \leadsto \false, \text{ if $\Iv \in \var (\hdg)$.}\\
{\sf (F2)} &\quad & & \Sv \eqls (\hdg_1,t,\hdg_2)  \leadsto  \false,   \text{ if $\Sv  \in  \var(H_1,t,H_2)$.}\\
{\sf (F3)} &\quad & & f_1(\hdg_1)\eqls f_2(\hdg_2)  \leadsto   \false, \text{ if $f_1 \not\myequiv f_2$.} \\
 {\sf (F4)} &\quad & & \epsilon \eqls (\hdg_1,t,\hdg_2) \leadsto \false. \\
{\sf (F5)}  &\quad & & f_1(\hdg) \inC f_2(\rh) \leadsto  \false,  \text{ if $f_1 \not\myequiv f_2$.} \\
{\sf (F6)} &\quad & & \epsilon \inC \rh  \leadsto  \false, \text{ if $\epsilon \not \in \regl{\rh}$}.\\
{\sf (F7)} &\quad & & (\hdg_1,t, \hdg_2) \inC \eps  \leadsto  \false. 
\vspace{-0.5cm}
\end{alignat*} 
\noindent\rule{0.9\textwidth}{0.4pt}\rule{0.1\textwidth}{4pt}

\subsubsection*{Decomposition Rules.} 
The set of these rules is denoted by {\sf Dec}. They operate on a conjunction of literals and give back either a conjunction of literals again, or a constraint in DNF. 

\vskip\medskipamount
\leaders\vrule width \textwidth \vskip0.4pt 
\leaders\vrule width 0.1\textwidth \vskip4pt 
\nointerlineskip
\begin{alignat*}{7}
{\sf (D1)} &\quad & &  \fu(\hdg)\eqls \fu(\tseq)  \land \conjc  
    \leadsto  \bigvee_{\tseq' \in perm(\tseq)}  \bigl(\hdg \eql \tseq'  \land \conjc\bigr ),\\
   & & & \text{where $\hdg$ and $\tseq$ are disjoint.}\\
{\sf (D2)} &\quad & &  (t_1,\hdg_1)  \eqls (t_2, \hdg_2) \leadsto  t_1 \eql t_2 \land \hdg_1 \eql \hdg_2, \text{ where } \hdg_1\neq \epsilon \text{ or } \hdg_2\neq \epsilon.
\vspace{-0.5cm}
\end{alignat*} 
\noindent\rule{0.9\textwidth}{0.4pt}\rule{0.1\textwidth}{4pt}
\subsubsection*{Deletion Rules.} 
These rules  delete identical terms or hedge variables from both sides of an equation. 
We denote this set of rules by {\sf Del}.
\vskip\medskipamount
\leaders\vrule width \textwidth \vskip0.4pt 
\leaders\vrule width 0.1\textwidth \vskip4pt 
\nointerlineskip
\begin{alignat*}{7}
{\sf (Del1)} &\quad & & (\Sv,\hdg_1)\eqls (\Sv,\hdg_2) \leadsto  \hdg_1\eql \hdg_2.\\
{\sf (Del2)} &\quad & & \fu(\hdg_1,h,\hdg_2)\eqls \fu(\hdg_3,h,\hdg_4)  \leadsto \fu(\hdg_1,\hdg_2)\eql \fu(\hdg_3,\hdg_4).\\
{\sf (Del3)} &\quad & &  \Sv\eqls (\hdg_1,\Sv,\hdg_2)   \leadsto \hdg_1\eql \epsilon \land \hdg_2\eql \epsilon, \text{ if } \hdg_1\neq \epsilon.
\vspace{-0.5cm}
\end{alignat*} 
\noindent\rule{0.9\textwidth}{0.4pt}\rule{0.1\textwidth}{4pt}
\subsubsection*{Variable Elimination Rules.} 
These rules eliminate variables from the given constraint keeping only a solved equation for them. They apply to disjuncts. 
The first two rules replace a variable with the corresponding expression, provided that the occurrence check fails:
\vskip\medskipamount
\leaders\vrule width \textwidth \vskip0.4pt 
\leaders\vrule width 0.1\textwidth \vskip4pt 
\nointerlineskip
\begin{alignat*}{7}
{\sf (E1)} &\quad & & \Iv\eqls t  \land \conjc  \leadsto
    \Iv\eql t\land \conjc \vartheta,\\
& & &\parbox{0.8\linewidth}{where $\Iv \not \in \var(t)$, $\Iv \in \var(\conjc)$ and $\vartheta =\{\Iv \mapsto t \}$. If $t$ is a variable then in addition it is required that $t \in \var(\conjc)$.}\\
{\sf (E2)} &\quad & & \Sv\eqls \hdg\land  \conjc   \leadsto \Sv\eql \hdg\land  \conjc \vartheta,\\
& & &\parbox{0.8\linewidth}{where  $\Sv \not \in \var(\hdg)$, $\Sv  \in \var(\conjc)$, and $\vartheta =\{\Sv \mapsto \hdg \}$. If $\hdg \myequiv \Svy$ for some $\Svy$, then in addition it is required that $\Svy  \in \var(\conjc)$.}
\vspace{-0.5cm}
\end{alignat*} 
\noindent\rule{0.9\textwidth}{0.4pt}\rule{0.1\textwidth}{4pt}

The next two rules {\sf (E3)} and {\sf (E4)} assign to a variable an initial part of the hedge in the other side of the selected equation. The hedge has to be a sequence of terms $T$ in the first rule. The disjunction in the rule is over all possible splits of $T$. In the second rule, only a split of the prefix $T$ of the hedge is relevant and the disjunction is over all such possible splits of $T$. The rest is blocked by the term $t$ due to occurrence check: No instantiation of $\Sv$ can contain it.

\vskip\medskipamount
\leaders\vrule width \textwidth \vskip0.4pt 
\leaders\vrule width 0.1\textwidth \vskip4pt 
\nointerlineskip
\begin{alignat*}{7}
{\sf (E3)} &\ && (\Sv,\hdg)\eqls \tseq\land  \conjc  \leadsto\quad \bigvee_{\tseq=(\tseq_1,\tseq_2)} \Bigl (\Sv \eql \tseq_1 \land \hdg \vartheta \eql \tseq_2 \land \conjc\vartheta \Bigl ),\\
& & &\text{where $\Sv\not\in\var (\tseq)$, $\vartheta=\{\Sv \mapsto \tseq_1\}$, and $\hdg\neq \epsilon$.}\\
{\sf (E4)} &\  & & (\Sv,\hdg_1)\eqls (\tseq,t,\hdg_2)\land  \conjc  \leadsto 
           \bigvee_{\tseq=(\tseq_1,\tseq_2)}\, \Bigl (\Sv \eql \tseq_1 \land  \hdg_1 \vartheta \eql (\tseq_2,t,\hdg_2) \vartheta \land \conjc\vartheta \Bigr )\\
& & & \text{where $\Sv\not\in\var (\tseq)$, $\Sv\in\var ( t)$, $\vartheta=\{\Sv \mapsto \tseq_1\}$, and $\hdg_1\neq \epsilon$.}
\vspace{-0.5cm}
\end{alignat*} 
\noindent\rule{0.9\textwidth}{0.4pt}\rule{0.1\textwidth}{4pt}

Finally, there are three rules for function variable elimination. Their behavior is standard:
   
\vskip\medskipamount
\leaders\vrule width \textwidth \vskip0.4pt 
\leaders\vrule width 0.1\textwidth \vskip4pt 
\nointerlineskip   
\begin{alignat*}{7}
 {\sf (E5)} &\quad & & \Fv\eqls \Gv \land \conjc \leadsto   \Fv\eql\Gv  \land \conjc \vartheta,  \\
            &      & & \parbox{0.8\linewidth}{ where $\Fv\not\myequiv \Gv$, $\Fv\in \var(\conjc)$, and $\vartheta= \{\Fv \mapsto \Gv \}$. If $\Gv$ is a function variable, then in addition it is required that $\Gv \in \var(\conjc)$.}\\
{\sf (E6)} &\quad && \Fv(\hdg_1 )\eqls \Gv(\hdg_2) \land \conjc \leadsto  \Fv\eql\Gv \land \Gv(\hdg_1)\vartheta \eql \Gv(\hdg_2)\vartheta  \land \conjc \vartheta. \\
           &      && \text{where  $\Fv\not\myequiv \Gv$, $\vartheta=\{\Fv \mapsto \Gv \}$, and $\hdg_1 \neq \epsilon$ or $\hdg_2 \neq \epsilon$.}\\
{\sf (E7)} &\quad & & \Fv(\hdg_1 )\eqls \Fv(\hdg_2) \land \conjc \leadsto \bigvee_{f\in\Fs} \Bigl(\Fv\eql f \land f(\hdg_1)\vartheta \eql f(\hdg_2)\vartheta  \land \conjc \vartheta \Bigr),\\ 
      & & &\text{where  $\vartheta=\{\Fv \mapsto f \}$, and $\hdg_1 \neq \hdg_2 $.}
\vspace{-0.5cm}
\end{alignat*} 
\noindent\rule{0.9\textwidth}{0.4pt}\rule{0.1\textwidth}{4pt}

We denote the set of rules {\sf (E1)}--{\sf (E7)} by {\sf Elim}. Note that the assumption of finiteness of $\Fs$ guarantees that the disjunction in {\sf (E7)} is finite. 
\subsubsection*{Membership Rules.} 
The membership rules apply to disjuncts of constraints in DNF, to preserve the DNF structure. They provide the membership check, if the hedge $\hdg$ in the membership atom $\hdg \inC \rh$ is ground. Nonground hedges require more special treatment as one can see.

To solve membership constraints for hedges of the form $(t,H)$ with $t$ a term, we rely on the possibility to compute the linear form of a regular expression, that is, to express it as a finite sum of concatenations of regular hedge expressions that identify all plausible membership constraints for $t$ and $H$.
Formally, the \emph{linear form} of a regular expression $\rh$, denoted $\lf(\rh)$, is a finite set of pairs  $(f(\rh_1), \rh_2)$, 
which is defined recursively as follows:
\begin{alignat*}{5}
 &    \lf(\eps) = {}  \emptyset.\\
 &    \lf(f(\rh)) = {}  \{( f(\rh), \eps ) \}. \\
 &    \lf(\rh_1 \chs \rh_2) = {}  \lf(\rh_1) \cup \lf(\rh_2).\\
 &    \lf(\rh_1\concs\rh_2) = {}  \lf(\rh_1) \odot \rh_2,  \text{ if } \epsilon\notin\regl{\rh_1}.\\
 &    \lf(\rh_1\concs\rh_2) = {}  \lf(\rh_1) \odot \rh_2 \cup\lf(\rh_2),  \text{ if } \epsilon\in\regl{\rh_1}.\\
 &    \lf(\rh^\strs) =   \lf(\rh) \odot \rh^\strs.  
\end{alignat*}
These equations involve an extension of concatenation $\odot$ that acts on a linear form and a regular expression and returns a linear form. It is defined as $l\odot \eps =  l,$ and $l\odot \rh  =  \{( f(\rh_1), \rh_2\concs\rh ) \mid ( f(\rh_1), \rh_2  )\in l,\rh_2\neq \eps\}\, \cup \{( f(\rh_1), \rh ) \mid ( f(\rh_1), \eps  )\in l\}$, if $\rh\neq \eps.$

The linear form $\lf(\rh)$ of a regular expression $\rh$ has the property~\cite{DBLP:journals/tcs/Antimirov96}:\footnote{In \cite{DBLP:journals/tcs/Antimirov96}, this property has been formulated for word regular expressions, but it straightforwardly extends to regular hedge expressions we use in this paper.} 
\begin{equation}\label{eq:lf}
 \regl{\rh}\setminus \{\epsilon\}= \bigcup_{(f(\rh_1),\rh_2)\in \lf(\rh)}\regl{f(\rh_1)\concs \rh_2},
\tag{\sc{lf}}
\end{equation}
which justifies its use in the rule {\sf M2} below. 

The first group of membership rules looks as follows:

\vskip\medskipamount
\leaders\vrule width \textwidth \vskip0.4pt 
\leaders\vrule width 0.1\textwidth \vskip4pt 
\nointerlineskip   
\begin{alignat*}{7}
{\sf (M1)} &\quad &   &  (\Sv_1,\ldots,\Sv_n) \inC \eps  \land \conjc \leadsto  \land_{i=1}^n\ \Sv_i \eql \epsilon \land  \conjc \vartheta, \\
 &      &   & \text{where }  \vartheta =\{\Sv_1 \mapsto \epsilon,\ldots,\Sv_n \mapsto \epsilon\}, n>0.\\
{\sf (M2)} &\quad &   & (t,\hdg) \inC  \rh\land \conjc  \leadsto \bigvee_{(f(\rh_1),\rh_2)\in \lf(\rh)} \Bigl(t \inC f(\rh_1) \land \hdg \inC \rh_2\land \conjc\Bigr),\\
           &      &   & \text{where $\hdg\neq\epsilon$ and $\rh \neq\eps $.}\\
{\sf (M3)} &\quad &   &(\Sv,\hdg) \inC f(\rh) \land \conjc   \leadsto  \\
           &      &   & \qquad  \Bigl( \Sv \inC f(\rh) \land  \hdg \eql \epsilon  \land \conjc \Bigr) \lor  \Bigl(\Sv \eql \epsilon  \land  \hdg \inC f(\rh)\land \conjc \Bigr),\\
           &      &   & \text{where $\hdg\neq\epsilon$.}\\
{\sf (M4)} &\quad &   & t \inC  \rh^ \strs  \leadsto  t \inC  \rh. \\
{\sf (M5)} &\quad &   & t \inC  \rh_1 \concs \rh_2\land \conjc  \leadsto 
            \Bigl(t \inC  \rh_1 \land \epsilon \inC \rh_2\land \conjc\Bigr) \lor  \Bigl(\epsilon \inC  \rh_1 \land t \inC \rh_2\land \conjc \Bigr). \\
{\sf (M6)} &\quad &   & t \inC  \rh_1 \chs \rh_2\land \conjc  \leadsto \Bigl(t \inC  \rh_1\land \conjc \Bigr) \lor  \Bigl( t \inC \rh_2\land \conjc \Bigr). \\
{\sf (M7)} &\quad &   & (\Sv,\hdg) \inC  \rh_1 \chs \rh_2\land \conjc  \leadsto  \Bigl( (\Sv,\hdg) \inC  \rh_1\land \conjc \Bigr)  \lor  \Bigl( (\Sv,\hdg) \inC \rh_2\land \conjc \Bigr). \\
{\sf \hphantom{1}(M8)} &\quad & & v \inC \rh_1 \land  v \inC \rh_2   \leadsto {} v \inC \rh,\\ 
& & & \text{ where } v\in \cVi\cup \cVs,  \regl{\rh} =  \regl{\rh_1}\cap \regl{\rh_2}, \text{ and neither } v \inC \rh_1 \text{ nor }  v \inC \rh_2 \\
 & & &  \text{can be transformed by the other rules}.
\vspace{-0.5cm}
\end{alignat*} 
\noindent\rule{0.9\textwidth}{0.4pt}\rule{0.1\textwidth}{4pt}

Next, we have rules which constrain singleton hedges to be in a term language. They proceed by the straightforward matching or decomposition of the structure. Note that in {\sf (M12)}, we require the arguments of the unordered function symbol to be terms. {\sf (M10)} and {\sf (M9)} do not distinguish whether $f$ is ordered or unordered:
\vskip\medskipamount
\leaders\vrule width \textwidth \vskip0.4pt 
\leaders\vrule width 0.1\textwidth \vskip4pt 
\nointerlineskip 
\begin{alignat*}{7}
{\sf (M9)} &\quad & & \Sv \inC f(\rh) \land \conjc  \leadsto {} \Sv \eql \Iv \land \Iv \inC f(\rh) \land \conjc\{\Sv\mapsto \Iv\}, \text{where $\Iv$ is fresh.}\\
{\sf (M10)} &\quad & &\Fv(\hdg) \inC f(\rh) \land \conjc    \leadsto  \Fv \eql f \land  f(\hdg) \{\Fv \mapsto f\} \inC f(\rh) \land   \conjc \{\Fv \mapsto f\}. \\ 
{\sf (M11)} &\quad & &\fo(\hdg) \inC \fo(\rh)     \leadsto {}   \hdg \inC \rh.\\
{\sf (M12)} &\quad & & \fu(\tseq) \inC \fu(\rh)\land \conjc    \leadsto {}  \bigvee_{\tseq' \in \perm( \tseq)}  \Bigl( \tseq' \inC \rh \land \conjc \Bigr).
\vspace{-0.5cm}
\end{alignat*} 
\noindent\rule{0.9\textwidth}{0.4pt}\rule{0.1\textwidth}{4pt}

We denote the set of rules {\sf (M1)}--{\sf (M12)} by {\sf Memb}.

\subsection{The Constraint Solving Algorithm}
\label{sect:algo}

In this section we present an algorithm that converts a constraint with respect to the rules specified in Section~\ref{subsect:rules} into a partially solved one. First, we define the rewrite step
\begin{equation*}
\label{step}        
        \step := \text{\sf first(Log, Fail, Del, Dec, Elim, Memb).}
 \end{equation*}

When applied to a constraint, $\step$ transforms it by the \emph{first} applicable rule of the solver, looking successively into the sets {\sf Log}, {\sf Fail}, {\sf Del}, {\sf Dec}, {\sf Elim}, and {\sf Memb}. If none of them apply, then the constraint is said to be in a \emph{normal form} with respect to $\step$.

The constraint solving algorithm implements the strategy $\sol$ defined as a repeated application of the rewrite step, aiming at the computation of a normal form with respect to $\step$. But it also makes sure that the constraint, passed to {\sf step}, is in DNF: 
\begin{equation*}
\sol:={\sf compose(dnf, NF(step)).}
\end{equation*}

Hence, $\sol$ takes a quantifier-free constraint, transforms it into its equivalent constraint in DNF (the strategy {\sf dnf} in the definition stands for the algorithm that does it), and then repeatedly applies $\step$ to the obtained constraint in DNF as long as possible. It remains to show that this definition yields an algorithm, which amounts to proving that the strategy {\sf NF(step)} indeed produces a constraint to which none of the rules from {\sf Log}, {\sf Fail}, {\sf Del}, {\sf Dec}, {\sf Elim}, and {\sf Memb} apply. The termination theorem states exactly this:

\begin{restatable}[Termination of $\sol$]{theorem}{twoThmTermination}
\label{thm:termination}
$\sol$ terminates on any quantifier-free constraint.
\end{restatable}

With the next two statements we show that the solver reduces a constraint to its equivalent constraint:
\begin{restatable}{lemma}{threeLemEquiv}
\label{lem:equivalence}
If $\step(\constr)=\constrd$, then $ \istruct \models \forall \Bigl(\constr \lra \overline{\exists}_{\var(\constr)}\constrd\Bigr)$ for all intended structures $\istruct$. 
\end{restatable} 
\begin{restatable}{theorem}{fourThmEquiv}
\label{thm:equivalence}
If $\sol(\constr)=\constrd$, then $\istruct \models \forall \Bigl(\constr \lra \overline{\exists}_{\var(\constr)}\constrd\Bigr)$ for all intended structures $\istruct$, and $\constrd$ is either partially solved or the $\false$ constraint. 
\end{restatable}

\section{Operational Semantics of $\clphc$}
\label{sect:operational}

In this section we  describe the operational semantics of $\clphc$, following the approach for the CLP schema given in \cite{DBLP:journals/jlp/JaffarMMS98}. A  \emph{state} is a pair $\lbr G\sep \constr\rbr$, where $G$ is the sequence of literals and $\constr = \conjc_1 \lor \cdots \lor \conjc_n$, where $\conjc$'s are conjunctions of $\true$, $\false$, and primitive constraints. The \emph{definition of an atom $p(t_1,\ldots,t_m)$ in program} $P$, $\defn_{\Pr}(p(t_1,\ldots,t_m))$, is the set of rules in $\Pr$ such that the head of each rule has  a form $p(r_1,\ldots,r_m)$. We assume that $\defn_{\Pr}$ each time returns fresh variants. 

A state $ \lbr L_1,\ldots, L_n \sep \constr \rbr $  can be \emph{reduced with respect to $P$} as follows: Select a literal $L_i$. Then:
\begin{itemize}
\item If $L_i$ is  a primitive constraint and $\sol(\constr \land L_i)\not=\false$, then it is reduced to $ \lbr L_1,\ldots, L_{i-1},L_{i+1},\ldots, L_n \sep \sol(\constr \land L_i)\rbr$.
\item If $L_i$ is  a primitive constraint and $\sol(\constr \land L_i)=\false$, then it is reduced to $ \lbr \boxmy \sep \false\rbr$.
\item If $L_i$ is an atom $p(t_1,\ldots,t_m)$, then it is reduced to
\begin{align*}
 \lbr L_1,\ldots, L_{i-1},t_1 \doteq r_1,\ldots,t_m \doteq r_m,B,L_{i+1},\ldots, L_n \sep \constr \rbr
\end{align*}
for some $(p(r_1,\ldots,r_m) \myrule B)\in \defn_{\Pr}(L_i)$.
\item If $L_i$ is a atom and $\defn_{\Pr}(L_i)=\emptyset$, then it is reduced to $ \lbr \boxmy \sep \false\rbr$.
\end{itemize}

A \emph{derivation from a state} $S$ in a program $\Pr$ is a finite or infinite sequence of states $S_0 \rightarrowtail S_1 \rightarrowtail \cdots \rightarrowtail S_n \rightarrowtail \cdots$ where $S_0$ is $S$ and there is a reduction from each $S_{i-1}$ to $S_i$, using rules in $\Pr$. A \emph{derivation from a goal} $G$ in a program $\Pr$ is a derivation from $\lbr G \sep \true \rbr $. The \emph{length} of a (finite) derivation of the form  $S_0 \rightarrowtail S_1 \rightarrowtail \cdots \rightarrowtail S_n$ is $n$. A derivation is \emph{finished} if the last goal cannot be reduced, that is, if its last state is of the form $ \lbr \boxmy \sep \constr\rbr$ where $\constr$ is partially solved or $\false$. If $\constr $ is $\false$, the derivation is said to be \emph{failed}. 

Naturally, it is interesting to find syntactic restrictions for programs guaranteeing that non-failed finished derivations produce a solved constraint instead of a partially solved one. In the next two sections we consider such restrictions, leading to well-moded and KIF style $\clphc$ programs that have the desired property.

\section{Well-Moded Programs}
\label{sect:well-moded}

The concept of well-modedness is due to~\cite{DBLP:conf/slp/DembinskiM85}. A \emph{mode} for an $n$-ary predicate symbol $p$ is a function $m_p:\{1,\ldots,n\}\longrightarrow\{\Input, \Output\}$. If $m_p(i)=\Input$ (resp. $m_p(i)=\Output)$ then the position $i$ is called an \emph{input} (resp. \emph{output}) \emph{position} of $p$. The predicates $\inC$ and $\eql$ have only output positions. For a literal $L=p(t_1,\ldots,t_n)$ (where $p$ can be also $\inC$ or $\eql$), we denote by $\invar(L)$ and $\outvar(L)$ the sets of variables occurring in terms in the input and output positions of $p$. 

If a predicate is used with different modes $m^1_p,\ldots, m^k_p$ in the program, we may consider each $p_{m^i_p}$ as a separate predicate. Therefore, we can assume without loss of generality that every predicate has exactly one mode (cf., e.g., \cite{DBLP:conf/ctrs/GanzingerW92}).

An \emph{extended literal} $E$ is either a literal, $\true$, or $\false$. We define $\invar(\true):=\emptyset$, $\outvar(\true):=\emptyset$, $\invar(\false):=\emptyset$, and $\outvar(\false):=\emptyset$.

A \emph{sequence of extended literals $E_1,\ldots, E_n$ is well-moded}  if the following hold:
\begin{enumerate}
\item For all $1\le i \le n$, $\invar(E_i) \subseteq \bigcup^{i-1}_{j=1}\outvar(E_j)$.
\item If for some $1\le i \le n$, $E_i$ is $t_1 \eql t_2$, then $\var(t_1) \subseteq \bigcup^{i-1}_{j=1}\outvar(E_j)$ or $\var(t_2) \subseteq \bigcup^{i-1}_{j=1}\outvar(E_j)$.
\item If for some $1\le i \le n$, $E_i$ is a membership atom, then the inclusion $\var(E_i) \subseteq \bigcup^{i-1}_{j=1}\outvar(E_j)$ holds.
\end{enumerate}

A \emph{conjunction  of extended literals $G$ is well-moded} if there exists a well-moded sequence of extended literals $E_1,\ldots,E_n$ such that $G=\bigwedge_{i=1}^n E_i$ modulo associativity and commutativity of conjunction.
A \emph{formula in DNF is well-moded} if each of its disjuncts is. A \emph{state $\langle L_1,\ldots, L_n \sep \conjc_1\lor \cdots \lor \conjc_m \rangle $ is well-moded}, where  $\conjc$'s are  conjunctions of $\true$, $\false$, and primitive constraints, if the formula $(L_1 \land\cdots \land L_n\land \conjc_1) \lor \cdots \lor (L_1 \land\cdots \land L_n \land \conjc_m)$ is well-moded.

A \emph{clause $A\myrule L_1,\ldots, L_n$ is  well-moded} if the following hold:
\begin{enumerate}
 \item \label{wel} For all $1\le i \le n$, $\invar(L_i)  \subseteq \bigcup^{i-1}_{j=1}\outvar(L_j)\cup \invar(A)$.
 \item $\outvar(A)  \subseteq \bigcup^{n}_{j=1}\outvar(L_j)\cup \invar(A)$.
 \item If for some $1\le i \le n$, $L_i$ is $t_1 \eql t_2$, then $\var(t_1)\subseteq \bigcup^{i-1}_{j=1}\outvar(L_j)\cup \invar(A)$ or $\var(t_2)\subseteq \bigcup^{i-1}_{j=1}\outvar(L_j)\cup \invar(A)$.
  \item  If for some $1\le i \le n$, $L_i$ is a membership atom, then  $\outvar(L_i) \subseteq \bigcup^{i-1}_{j=1}\outvar(L_j)$ $\cup$ $\invar(A)$.
\end{enumerate}

A \emph{program is well-moded} if all its clauses are well-moded. 
\begin{example}
In Example~\ref{exmp:rewriting}, if in the user-defined binary predicates $\mathit{rewrite}$ and $\mathit{rule}$ the first argument is the input position and the second argument is the output position, then it is easy to see that the program is well-moded. In Example~\ref{exmp:rpo}, for well-modedness we need to define both positions in the user-defined predicates to be the input ones.
\end{example}

In the rest of this section we investigate the behavior of well-moded programs. Before going into the details, we briefly summarize two main results:

\begin{itemize}
 \item The solver can completely solve satisfiable well-moded constraints (instead of partial solutions computed in the general case). See Theorem~\ref{lem:solvedform}.
 \item Any finished derivation from a well-moded goal with respect to a well-moded program either ends with a completely solved constraint, or fails. See Theorem~\ref{thm:finalstate}.
\end{itemize}

To prove these statements, some technical lemmas are needed. 
\begin{restatable}{lemma}{fiveLemWellModedness}
  \label{lem:well:modedness:1}
  Let $v \eql e$ be an equation, where $v$ is a variable and $e$ is the corresponding expression such that $v$ does not occur in $e$. Let $\conjc_1$ and $\conjc_2$ be two arbitrary (possibly empty) conjunctions of extended literals such that the conjunction $\conjc_1\land \conjc_2 \land v \eql e$ is well-moded.  Let $\substb=\{v\mapsto e\}$ be a substitution. Then $\conjc_1\land \conjc_2\substb \land v \eql e$ is also well-moded.
\end{restatable}

The next lemma states that reduction with respect to a well-moded program preserves well-modedness of states:
\begin{restatable}{lemma}{sixLemWellModednessReduction}
\label{lem:well:modedness:2}
 Let $\Pr$ be a well-moded $\clphc$ program and $\lbr \Qr \sep \constr \rbr $ be a well-moded state. If $\lbr \Qr \sep \constr \rbr \rightarrowtail \lbr \Qr' \sep \constr' \rbr $ is a reduction using clauses in $\Pr$, then $\lbr \Qr' \sep \constr' \rbr $ is also a well-moded state.
\end{restatable}
\begin{restatable}{corollary}{sevenCorWellModedness}
\label{cor:wellmoded}
If $\constr $ is a well-moded constraint, then $\sol(\constr)$ is also well-moded.
\end{restatable}

The following theorem shows that satisfiable well-moded constraints can be completely solved:
\begin{restatable}{theorem}{eightThmSolvedForm}
\label{lem:solvedform}
Let $\constr $ be a well-moded constraint and $\sol(\constr)= \constr'$, where $\constr'\neq \false$. Then $\constr'$ is solved. 
\end{restatable}

We illustrate how to solve a simple well-moded constraint:

\begin{example} Let $\constr = f(\Sv, a,\Svy)\eql f(a,b,a,c,c) \land f(\Svz,a,\Iv) \eql f(\Svy,\Sv) \land \Svy \inC c(\eps)^\strs$. Then $\sol$ performs the following derivation (some steps are contracted):
\begin{alignat*}{7}
\constr &\leadsto {} 
    & & \bigl (\Sv \eql \epsilon \land (a,\Svy) \eql (a,b,a,c,c) \land f(\Svz,a,\Iv) \eql f(\Svy,\Sv) \land \Svy \inC c(\eps)^\strs \bigr )  \\
    & & \lor {}& \bigl (\Sv \eql a \land (a,\Svy) \eql (b,a,c,c) \land f(\Svz,a,\Iv) \eql f(\Svy,\Sv) \land \Svy \inC c(\eps)^\strs \bigr )\\
    & & \lor {}&\bigl (\Sv \eql (a,b) \land (a,\Svy) \eql (a,c,c) \land f(\Svz,a,\Iv) \eql f(\Svy,\Sv) \land \Svy \inC c(\eps)^\strs \bigr ) \\
    & & & \cdots \\
    & & \lor {}  &\bigl (\Sv \eql (a,b,a,c,c) \land (a,\Svy) \eql \epsilon \land f(\Svz,a,\Iv) \eql f(\Svy,\Sv) \land \Svy \inC c(\eps)^\strs \bigr ) \\ 
 &\leadsto^+ {}  & &\bigl (\Sv \eql \epsilon \land \Svy \eql (b,a,c,c) \land f(\Svz,a,\Iv) \eql f(\Svy,\Sv) \land \Svy \inC c(\eps)^\strs \bigr )  \\
   & & \lor {} &\bigl (\Sv \eql (a,b) \land \Svy \eql (c,c) \land f(\Svz,a,\Iv) \eql f(\Svy,\Sv) \land \Svy \inC c(\eps)^\strs \bigr )  \\
 &\leadsto {}  & &\bigl (\Sv \eql \epsilon \land \Svy \eql (b,a,c,c) \land f(\Svz,a,\Iv) \eql f(b,a,c,c,\Sv) \land (b,a,c,c) \inC c(\eps)^\strs \bigr )  \\
   & & \lor {} &\bigl (\Sv \eql (a,b) \land \Svy \eql (c,c) \land f(\Svz,a,\Iv) \eql f(\Svy,\Sv) \land \Svy \inC c(\eps)^\strs \bigr )  \\
 &\leadsto {}  & &\bigl (\Sv \eql \epsilon \land \Svy \eql (b,a,c,c) \land f(\Svz,a,\Iv) \eql f(b,a,c,c,\Sv) \\
       &    & & \qquad \land b \inC c(\eps) \land (a,c,c) \inC c(\eps)^\strs \bigr )  \\
   & & \lor {} &\bigl (\Sv \eql (a,b) \land \Svy \eql (c,c) \land f(\Svz,a,\Iv) \eql f(\Svy,\Sv) \land \Svy \inC c(\eps)^\strs \bigr )  \\
 & \leadsto {}  & &\bigl (\Sv \eql (a,b) \land \Svy \eql (c,c) \land f(\Svz,a,\Iv) \eql f(\Svy,\Sv) \land \Svy \inC c(\eps)^\strs \bigr )  \\
 & \leadsto^+ {}  && \bigl (\Sv \eql (a,b) \land \Svy \eql (c,c) \land f(\Svz,a,\Iv) \eql f(c,c,a,b) \land (c,c) \inC c(\eps)^\strs \bigr )  \\
  &  \leadsto^+ {}  && \bigl (\Sv \eql (a,b) \land \Svy \eql (c,c) \land f(\Svz,a,\Iv) \eql f(c,c,a,b) \bigr )  \\
  &  \leadsto^+ {}  && \bigl (\Sv \eql (a,b) \land \Svy \eql (c,c) \land \Svz \eql (c,c) \land \Iv \eql b \bigr ).
\end{alignat*}
The obtained constraint is solved.
\end{example}

The next theorem is the main result for well-moded $\clphc$ programs. It states that any finished derivation from a well-moded goal leads to a solved constraint or to a failure:
\begin{restatable}{theorem}{nineThmMainWellModed}
 \label{thm:finalstate}
Let $ \lbr \Qr \sep \true \rbr \rightarrowtail \cdots \rightarrowtail   \lbr \boxmy \sep \constr\rbr$ be a finished derivation with respect to a well-moded $\clphc$ program, starting from a well-moded goal $\Qr$. If $\constr\neq \false$, then $\constr$ is solved.
\end{restatable}

\section{Programs in the KIF Form}
\label{sect:kif}

Knowledge Interchange Format, shortly KIF~\cite{genesereth92knowledge}, is a computer-oriented language for the interchange of knowledge among disparate programs. It permits variadic syntax and hedge variables, under the restriction that such variables are only the last arguments of subterms they appear in. Such a fragment has some good computation properties, e.g., unification is unitary~\cite{DBLP:conf/rta/Kutsia03}. The special form of programs and constraints considered in this section originates from this restriction. 

\emph{Terms and hedges in the KIF form} or, shortly, \emph{KIF terms} and \emph{KIF hedges}, are defined by the following grammar:
\begin{alignat*}{3}
 t_\kappa & ::= \Iv \mid \fo(\hdg_\kappa)  \mid \fu({t_\kappa}_1,\ldots,{t_\kappa}_n) \mid  \Fv({t_\kappa}_1,\ldots,{t_\kappa}_n)  \quad (n\ge 0) & &\text{KIF Term}\\
 \hdg_\kappa & ::= {t_\kappa}_1,\ldots,{t_\kappa}_n  \mid {t_\kappa}_1,\ldots,{t_\kappa}_n,\Sv \quad (n\ge 0) &  \qquad & \text{KIF Hedge}
\end{alignat*}

That means that a term is in the KIF form if hedge variables occur only below ordered function symbols as the last arguments. For example, the terms $\fo(\Iv,\fo(a,\Sv),\fu(\Iv,b),\Sv)$ and $\fo(a,\Iv,b)$ are in the KIF form, while $\fo(\Sv, a,\Sv)$ and $\fu(\Iv,\fo(a,\Sv),\fu(\Iv,b),\Sv)$ are not.

If the language does not contain unordered function symbols, then we permit hedge variables under function variables, again in the last position, i.e., of the form $\Fv(\hdg_\kappa)$.

In this section we consider only KIF terms. Therefore, the subscript $\kappa$ will be omitted. 

\emph{KIF equations} and \emph{KIF atoms} are constructed from KIF terms. In a \emph{KIF membership atom} $\hdg \inC \rh $, the hedge $\hdg$ is a KIF hedge.

\emph{KIF formulas} are constructed from KIF primitive constraints and KIF atoms. This special form guarantees that the solver does not need to use all the rules. Simply inspecting them, we can see that {\sf Del1}, {\sf E3}, {\sf E4}, and {\sf M3} are not used. In {\sf Del3}, it is guaranteed that $\hdg_2$ will be always empty, and in {\sf M1} the $n$ will be equal to 1. 

Similarly to the well-moded restriction above, our interest to the KIF fragment is justified by its two important properties that characterize the KIF constraint solving and derivation of KIF goals:
\begin{itemize}
  \item The solver can completely solve satisfiable KIF constraints (instead of partial solutions computed in the general case). See Theorem~\ref{thm:kif:solved}.
  \item Any finished derivation from a KIF goal with respect to a KIF program either ends with a completely solved constraint, or fails. See Theorem~\ref{thm:kif:derivation}.
\end{itemize}

Their proofs are easier than the ones of the corresponding statements for well-moded programs. This is largely due to the following lemma:
\begin{restatable}{lemma}{tenLemKIFPartiallySolved}
  \label{lem:kif:partially:solved}
  Any partially solved KIF constraint is solved.
\end{restatable}

One can see that no solving rule inserts a term or a hedge variable after the last argument of subterms in constraints. That means, KIF constraints are again transformed into KIF constraints. Hence, the constraint computed by $\sol$ will be a KIF constraint. It leads us to the following result:
\begin{restatable}{theorem}{elevenThmKIFSolved}
  \label{thm:kif:solved}
  Let $\constr $ be a KIF constraint and $\sol(\constr)= \constr'$, where $\constr'\neq \false$. Then $\constr'$ is solved. 
\end{restatable}

We illustrate now how to solve a simple KIF constraint:

\begin{example} Let $\constr = f(x,\Sv)\eql f(g(\Svy), a, \Svy) \land \Sv \inC a(\eps)^\strs \land \Svy \inC  a(\eps)\concs a(b(\eps)^\strs)^\strs$. Then $\sol$ performs the following derivation:
\begin{alignat*}{1}
\constr \leadsto {} & x \eql g(\Svy) \land \Sv \eql (a, \Svy) \land \Sv \inC a(\eps)^\strs \land \Svy \inC  a(\eps)\concs a(b(\eps)^\strs)^\strs\\
  \leadsto {} & x \eql g(\Svy) \land \Sv \eql (a, \Svy) \land (a, \Svy) \inC a(\eps)^\strs \land \Svy \inC  a(\eps)\concs a(b(\eps)^\strs)^\strs\\
         \leadsto {} & x \eql g(\Svy) \land \Sv \eql (a, \Svy) \land \Svy \inC a(\eps)^\strs \land \Svy \inC  a(\eps)\concs a(b(\eps)^\strs)^\strs\\
         \leadsto {} & x \eql g(\Svy) \land \Sv \eql (a, \Svy) \land \Svy \inC  a(\eps)\concs a(\eps)^\strs
\end{alignat*}
The obtained constraint is solved.
\end{example}

A \emph{state $\langle L_1,\ldots, L_n \sep \conjc_1\lor \cdots \lor \conjc_m \rangle $ is in the KIF form} (\emph{KIF state}), if the formula $(L_1 \land\cdots \land L_n\land \conjc_1) \lor \cdots \lor (L_1 \land\cdots \land L_n \land \conjc_m)$ is a KIF formula.

\emph{KIF clauses} are constructed from KIF atoms and literals. \emph{KIF programs} are sets of KIF clauses. It is not hard to check that each reduction step (with respect to a KIF program) in the operational semantics preserves KIF states: It follows from the definition of the operational semantics and the fact that $\sol$ computes KIF constraints. Therefore, we can establish the following theorem:
\begin{restatable}{theorem}{twelveThmKIFDerivation}
  \label{thm:kif:derivation}
  Let $ \lbr \Qr \sep \true \rbr \rightarrowtail \cdots \rightarrowtail   \lbr \boxmy \sep \constr'\rbr$ be a finished derivation with respect to a KIF program, starting from a KIF goal $\Qr$. If $\constr'\neq \false$, then $\constr'$ is solved.
\end{restatable}

\begin{example}
  \label{exmp:dl}
  The well-known technique of appending two difference lists can be used in $\clphc$ for a more general task: to combine arguments of arbitrary two terms. The program remains the same as in the standard logic programming:
\begin{alignat*}{1}
 \mathit{append\_dl}(x_1\mhyphen x_2,\, x_2\mhyphen x_3,\, x_1 \mhyphen x_3),
\end{alignat*}
where the hyphen is a function symbol and $x_1,x_2,x_3$ are term variables. The KIF goal 
 \begin{alignat*}{1}
 \mathit{append\_dl}(f_1(a,b,\Sv)\mhyphen f_2(\Sv),\, f_2(c,d,e,\Svy)\mhyphen f_3(\Svy),\, x \mhyphen f_3 )
\end{alignat*}
can be used to append to the arguments of $f_1(a,b)$ the arguments of $f_2(c,d,e)$, obtaining $f_1(a,b,c,d,e)$. Note that the terms may have different heads. The derivation proceeds as follows:
\begin{alignat*}{1}
  & \lbr \mathit{append\_dl}(f_1(a,b,\Sv)\mhyphen f_2(\Sv),\, f_2(c,d,e,\Svy)\mhyphen f_3(\Svy),\, x \mhyphen f_3 ) \sep \true\rbr \\
\rightarrowtail \quad & \lbr x_1\mhyphen x_2 \eql f_1(a,b,\Sv)\mhyphen f_2(\Sv),\, x_2\mhyphen x_3 \eql f_2(c,d,e,\Svy)\mhyphen f_3(\Svy),\, x_1 \mhyphen x_3 \eql x \mhyphen f_3 \sep \true \rbr\\
\rightarrowtail \quad & \lbr x_2\mhyphen x_3 \eql f_2(c,d,e,\Svy)\mhyphen f_3(\Svy),\, x_1 \mhyphen x_3 \eql x \mhyphen f_3 \sep x_1\eql f_1(a,b,\Sv) \land  x_2 \eql f_2(\Sv) \rbr \\
\rightarrowtail \quad & \lbr x_1 \mhyphen x_3 \eql x \mhyphen f_3 \sep \\
 & \qquad x_1\eql f_1(a,b,c,d,e,\Svy) \land  x_2 \eql f_2(c,d,e,\Svy) \land x_3 \eql f_3(\Svy) \land \Sv \eql (c,d,e,\Svy) \rbr\\
\rightarrowtail  \quad & \lbr \boxmy \sep \\
 & \qquad x_1\eql f_1(a,b,c,d,e) \land  x_2 \eql f_2(c,d,e) \land x_3 \eql f_3 \land \Sv \eql (c,d,e) \land \Svy \eql \epsilon \land {}\\
 & \qquad x\eql f_1(a,b,c,d,e) \rbr.
\end{alignat*}
The constraint in the final state is solved.
\end{example}

\section{Conclusion}
\label{sect:conclusion}

Solving equational and membership constraints over hedges is not an easy task: The problem is infinitary and any procedure that explicitly computes all solutions is non-terminating. The solver that we presented in this paper is not complete, but it is terminating. It solves constraints partially and tries to detect failure as early as it can. 

Incorporating the solver into the CLP schema gives $\clphc$: constraint logic programming for hedges. We defined algebraic semantics for it and used it to characterized the constraint solver: The output of the solver (which is either partially solved of $\false$) is equivalent to the input constraint in all intended structures. 

The fact that the solver, in general, returns a partially solved result (when it does not fail), naturally raises the question: Are there some interesting fragments of constraints that the solver can completely solve? We give a positive answer to this question, defining well-moded and KIF constraints and showing their complete solvability. 

It immediately poses the next question: Can one characterize $\clphc$ programs that generate only well-moded or KIF constraints only? We show that by extending the notions of well-modedness and KIF form to programs, we get the desired fragments. Any finished derivation of a goal for such fragments gives a definite answer: Either the goal fails, or a solved constraint is returned.

The constraints we consider in this paper are positive, but at least the well-moded programs can be easily enriched with the negation. Well-modedness guarantees that the eventual test for disequality or non-membership in constraints will be performed on ground hedges, which can be effectively decided.

\section*{Acknowledgments}

This research has been partially supported by LIACC through Programa de Financiamento Plurianual of the Funda\c{c}\~{a}o para a Ci\^{e}ncia e Tecnologia (FCT), by the FCT fellowship (ref. SFRH/BD/62058/2009), by the Austrian Science Fund (FWF) under the project SToUT (P 24087-N18), and the by Rustaveli Science Foundation under the grants DI/16/4-120/11 and FR/611/4-102/12.

\bibliographystyle{acmtrans}

\newpage
\appendix

\section{Proofs}

\oneThmSatisfiable*

\begin{proof}
Since $\constr$ is solved, each disjunct $\conjc$ in it has a form $v_1\eql e_1\land\cdots\land v_n\eql \penalty100000 e_n\land\allowbreak v'_1 \inC \rh_1\land\cdots\land v'_m \inC \rh_m$ where  $m,n\ge0$,  $v_i,v'_j\in \cVv$  and $e_i$ is an expression corresponding to $v_i$. Moreover, $v_1,\ldots,v_n,v'_1,\ldots,v'_m$ are distinct and $\regl{\rh_j}\neq \emptyset $ for all $1 \le j \le m$. Note that while $v_i$'s do not occur anywhere else in $\conjc$, it still might be the case that some $v'_j$, $1\le j\le m$, occurs in some $e_k$, $1\le k \le n$.

Let $e'_j$ be an element of $\regl{ \rh_j}$  for all $1 \le j \le m$. Assume that for each $1 \le i \le n$, the substitution $\sigma'_i$ is a grounding substitution for $e_i$ with the property that $v'_j\sigma'_i = e'_j$ for all $1\le j \le m$. Then $\sigma = \{ v_1\mapsto e_1\sigma'_1,\ldots,v_n\mapsto e_n\sigma'_n, v'_1 \mapsto e'_1,\ldots,v'_m \mapsto e'_m\}$ solves $\conjc$. Therefore, $\istruct\models \exists \constr $ holds.
\end{proof}

\twoThmTermination*

\begin{proof}
We need to show that {\sf NF(step)} terminates for any quantifier-free constraint in DNF. We define a complexity measure  $\cm(\constr)$ for such constraints, and show that $\cm(\constr')<\cm(\constr)$ holds whenever $\constr'=\text{\sf step}(\constr)$. 

For a hedge $\hdg$ (resp., for a regular expression $\rh$), we denote by $\size(\hdg)$ (resp., by $\size(\rh)$) its denotational length, e.g., $\size(\epsilon)=0$,  $\size(\eps)=1$, $ \size(f(f(a)),\Sv)=4$, and $ \size(f(f( a \concs b^\strs))) = 6$.  

The complexity measure $\cm(\conjc)$ of a conjunction of primitive constraints $\conjc$ is the tuple $\langle N_1, M_1,N_2,M_2,M_3 \rangle$ defined as follows ($\mset{}$ stands for a multiset):
\begin{itemize}
 \item $N_1$  is the number of unsolved variables in $\conjc$.
 \item $M_1 := \mset{\size(\hdg) \mid\hdg \inC \rh \in \conjc, \hdg\neq\epsilon}$.
 \item $N_2$ is the number of  primitive constraints in the form $\Sv \inC \rh$ in $\conjc$ .
  \item $M_2 := \mset{\size(\rh) \mid\hdg \inC \rh \in \conjc}$.
 \item $M_3 := \mset{\size(t_1)+\size(t_2) \mid t_1 \eql t_2 \in \conjc}$.
\end{itemize}

The complexity measure $\cm(\constr)$ of a constraint $\constr = \conjc_1 \lor \cdots \lor \conjc_n$ is defined as $\mset{\cm(\conjc_1),\ldots, \cm(\conjc_n)}$. 

Measures are compared by the multiset extension of the lexicographic ordering on tuples. The components that are natural numbers ($N_1$ and $N_2$) are, of course, compared by the standard ordering on naturals. The multiset components $M_1$, $M_2$, and $M_3$ are compared by the multiset extension of the standard ordering on the naturals.  

The strict part of the ordering on measures is obviously well-founded. The {\sf Log} rules strictly reduces it. For the other rules, the table below shows which rule reduces which component of the measure. The symbols $>$ and $\ge$ indicate the strict and non-strict decrease, respectively. It implies the termination of the algorithm $\sol$.

\begin{center}
 \begin{tabular}{|lc|c|c|c|c|c|}
  \hline
  Rule & \qquad & $N_1$ & $M_1$ & $N_2$ & $M_2$ & $M_3$  \\
  \hline \hline
  {\sf (M1), (M10), (E1)--(E7)} &  & $>$ & & & &\\
  {\sf (F5), (F7), (M2), (M3), (M8), (M11), (M12)} & & $\ge$ &  $>$ & & &\\
  {\sf  (M9)} & & $\ge$ &  $\ge$ & $>$ & &\\
  {\sf (F6), (M4)--(M7)} &&  $\ge $ &  $\ge$&  $\ge$  & $>$ &\\
  {\sf (D1), (D2), (F1)--(F4), (Del1)--(Del3)} & & $\ge$ & $\ge$ &  $\ge$ &  $\ge$ & $>$  \\ \hline
 \end{tabular} 
\end{center}
\end{proof}

\threeLemEquiv*

\begin{proof}
By case distinction on the inference rules of the solver, selected by the strategy {\sf first} in the application of {\sf step}. We illustrate here two cases, when the selected rules are {\sf (E3)} and {\sf (M2)}. For the other rules the lemma can be shown similarly.
 
In {\sf (E3)}, $\constr$ has a disjunct $\conjc = (\Sv,\hdg)\eql \tseq\land  \conjc'$ with $\Sv\not\in\var (\tseq)$, and $\constrd$ is the result of replacing $\conjc$ in $\constr$ with the disjunction $\constr'=\bigvee_{\tseq=(\tseq_1,\tseq_2)} (\Sv \eql \tseq_1 \land \hdg \vartheta \eql \tseq_2 \land \conjc'\vartheta)$ where $\vartheta=\{\Sv \mapsto \tseq_1\}$.  Therefore, it is sufficient to show that $\istruct\models \forall (\conjc \lra \overline{\exists}_{\var(\constr)}\constr')$. Since $\var(\constr')=\var(\conjc)$, this amounts to showing that for all ground substitutions $\sigma$ of $\var(\conjc)$ we have $\istruct\models (\Sv \sigma,\hdg \sigma) \eql \tseq \sigma \land  \conjc' \sigma$ iff $\istruct\models (\bigvee_{\tseq=(\tseq_1,\tseq_2)} (\Sv \eql \tseq_1 \land \hdg \vartheta \eql \tseq_2 \land \conjc'\vartheta))\sigma.$

\begin{itemize}
\item Assume $\istruct\models (\Sv \sigma,\hdg \sigma) \eql \tseq \sigma \land  \conjc' \sigma$. We can split  $\tseq \sigma $ into $\tseq_1 \sigma $ and $\tseq_2 \sigma $ such that $\Sv \sigma = \tseq_1 \sigma $ and $\hdg \sigma = \tseq_2 \sigma $. Now, we show $v \vartheta \sigma = v \sigma $ for all $v \in \var(\Sv,\hdg, \tseq)$. Indeed, if $v\neq \Sv$, the equality trivially holds. If $v=\Sv$, we have $\Sv\vartheta \sigma = \tseq_1\sigma=\Sv\sigma$. Hence, $\istruct\models (\bigvee_{\tseq=(\tseq_1,\tseq_2)} (\Sv \eql \tseq_1 \land \hdg \vartheta \eql \tseq_2 \land \conjc'\vartheta))\sigma.$

\item Assume $\istruct\models (\bigvee_{\tseq=(\tseq_1,\tseq_2)} (\Sv \eql \tseq_1 \land \hdg \vartheta \eql \tseq_2 \land \conjc'\vartheta))\sigma$. Then there exists the split $\tseq=(\tseq_1,\tseq_2)$ such that $\istruct\models (\Sv \sigma \eql \tseq_1 \sigma \land \hdg \vartheta \sigma \eql \tseq_2 \sigma \land \conjc'\vartheta \sigma)$. Again, we can show $v \vartheta \sigma = v \sigma $ for all $v \in \var(\Sv,\hdg, \tseq)$. Hence, $\istruct\models (\Sv \sigma,\hdg \sigma) = \tseq \sigma \land  \conjc' \sigma$. It finishes the proof for {\sf (E3)}.
\end{itemize}

Now, let the selected rule be {\sf (M2)}. In this case $\constr$ has a disjunct $\conjc = (t,\hdg) \inC  \rh\land \conjc'$ with $\hdg\neq\epsilon$ and $\rh \neq\eps $.  Then $\constrd$ is the result of replacing $\conjc$ in $\constr$ with $\constr'=\bigvee_{(f(\rh_1),\rh_2)\in \lf(\rh)}(t \inC f(\rh_1) \land \hdg \inC \rh_2\land \conjc')$. Therefore, to show  $\istruct\models \forall (\constr \lra \overline{\exists}_{\var(\constr)}\constrd)$, it is enough to show that  $\istruct\models \forall (\conjc \lra \overline{\exists}_{\var(\constr)}\constr')$. Since $\var(\constr')=\var(\conjc)$, this amounts to showing that for all ground substitutions $\sigma$ of $\var(\conjc)$ we have $\istruct\models (t \sigma,\hdg \sigma) \inC  \rh\land \conjc'\sigma $ iff $\istruct\models (\bigvee_{(f(\rh_1),\rh_2)\in \lf(\rh)}(t \inC f(\rh_1) \land \hdg \inC \rh_2\land \conjc'))\sigma$.

\begin{itemize}
\item Assume $\istruct\models (t \sigma,\hdg \sigma) \inC  \rh\land \conjc'\sigma $. By the property (\ref{eq:lf}) above and by the definitions of intended structure and entailment, we get that $\istruct\models (t \sigma,\hdg \sigma) \inC  \rh\land \conjc '\sigma $ implies $\istruct\models (t \sigma,\hdg \sigma) \inC  \lf(\rh) \land \conjc '\sigma $ . Hence, we can conclude $\istruct\models (\bigvee_{(f(\rh_1),\rh_2)\in \lf(\rh)}(t \sigma \inC f(\rh_1) \land \hdg \sigma \inC \rh_2\land \conjc'\sigma))$.

\item Assume $\istruct\models (\bigvee_{(f(\rh_1),\rh_2)\in \lf(\rh)}(t \sigma \inC f(\rh_1) \land \hdg \sigma \inC \rh_2\land \conjc'\sigma))$. Then we have $\istruct\models (t \sigma,\hdg \sigma) \inC  \lf(\rh) \land \conjc'\sigma$ which, by (\ref{eq:lf}), implies  $\istruct\models (t \sigma,\hdg \sigma) \inC  \rh\land \conjc'\sigma$.
\end{itemize}
\end{proof} 

\fourThmEquiv*

\begin{proof} 
We assume without loss of generality that $\constr$ is in DNF. 
$\istruct \models \forall \bigl(\constr \lra \overline{\exists}_{\var(\constr)}\constrd\bigr)$ follows from Lemma~\ref{lem:equivalence} and the following property: If $\istruct \models \forall \bigl(\constr_1 \lra \overline{\exists}_{\var(\constr_1)}\constr_2\bigr)$ and $\istruct \models \forall \bigl(\constr_2 \lra \overline{\exists}_{\var(\constr_2)}\constr_3\bigr)$, then $\istruct \models \forall \bigl(\constr_1 \lra \overline{\exists}_{\var(\constr_1)}\constr_3\bigr)$. The property itself relies on the fact that $\istruct \models \forall \bigl(\overline{\exists}_{\var(\constr_1)} \overline{\exists}_{\var(\constr_2)} \constr_3 \lra \overline{\exists}_{\var(\constr_1)}\constr_3\bigr)$, which holds because all variables introduced by the rules of the solver in $\constr_3$ are fresh not only for $\constr_2$, but also for $\constr_1$.

As for the partially solved constraint, by the definition of $\sol$ and Theorem~\ref{thm:termination}, $\constrd$ is in a normal form. Assume by contradiction that it is not partially solved. By inspection of the solver rules, based on the definition of partially solved constraints, we can see that there is a rule that applies to $\constrd$. But this contradicts the fact that $\constrd$ is in a normal form. Hence, $\constrd$ is partially solved. 
\end{proof}

\fiveLemWellModedness*
\begin{proof}
  The point in this lemma is that it does not matter how $\conjc_1$ and $\conjc_2$ are chosen.  We consider two cases. First, when $v \eql e$ is the leftmost literal containing $v$ in a well-moded sequence corresponding to $\conjc_1\land \conjc_2 \land v \eql e$ and, second, when this is not the case.

  \emph{Case 1.} Let $\se_1,v \eql e,\se_2$ be a well-moded sequence corresponding to $\conjc_1\land \conjc_2 \land v \eql e$, such that $\se_1$ does not contain $v$. Note that there is no assumption (apart from what guarantees well-modedness of $\conjc_1\land \conjc_2 \land v \eql e$) on the appearance of literals in $\se_1$ and $\se_2$: They may contain literals from $\conjc_1$ only, from $\conjc_2$ only, or from both $\conjc_1$ and $\conjc_2$. 

Well-modedness of $\se_1,v \eql e,\se_2$ requires the variables of $e$ to appear in $\se_1$. Consider the sequence $\se_1,v \eql e,\se_2[\substb]$, where the notation $\se[\substb]$ stands for such an instance of $\se$ in which $\substb$ affects only literals from $\conjc_2$. Then $\se_1,v \eql e$ is well-moded and it can be safely extended by $\se_2[\substb]$ without violating well-modedness, because the variables in $v\eql e$ still precede (in the well-moded sequence) the literals from $\se_2[\substb]$, and the relative order of the other variables (in the well-moded sequence) does not change. Hence, $\se_1,v \eql e,\se_2[\substb]$ is a well-moded sequence that corresponds to $\conjc_1 \land \conjc_2\substb \land v \eql e$.

  \emph{Case 2.} Let $\se_1,\lit,\se_2,v \eql e,\se_3$ be a well-moded sequence corresponding to $\conjc_1 \land \conjc_2 \land v \eql e$, where $\lit$ is the leftmost literal that contains $v$ in an output position. Again, we make no assumption on literal appearances in the subsequences of the sequence. Then $\se_1,\lit,v \eql e,\se_2,\se_3$ is also a well-moded sequence (corresponding to $\conjc_1 \land \conjc_2 \land v \eql e$), because $v$ still appears in an output position in $\lit$ left to $v \eql e$, the variables in $e$ still precede literals from $\se_3$, and the relative order of the other variables does not change. For literals in $\se_2$ that contain variables from $e$ such a reordering does not matter.

  Note that $v$ does not appear in $\se_1$: If it were there in some literal in an output position, then $\lit$ would not be the leftmost such literal. If it were there in some literal $\lit'$ in an input position, then well-modedness of the sequence would require $v$ to appear in an output position in another literal $\lit''$ that is even before $\lit'$, i.e., to the left of $\lit$ and it would again contradict the assumption that $\lit$ is the leftmost literal containing $v$ in an output position.

  Let $\se_1,\lit[\substb],v\eql e, \se_2[\substb],\se_3[\substb]$ be a sequence of all literals taken from $\conjc_1 \land \conjc_2 \land v \eql e$. We distinguish two cases, depending whether $\substb$ affects $\lit$ or not.
  \begin{description}
  \item[\rm\emph{$\substb$ affects $\lit$.}] Then it replaces $v$ in $\lit$ with $e$, i.e., $\lit[\substb]=\lit\substb$. Then the variables of $e$ appear in output positions in $\lit\substb$ and, hence, placing $v\eql e$ after  $\lit\substb$ in the sequence would not destroy well-modedness. As for the $\lit\substb$ itself, we have two alternatives: 
  \begin{enumerate}
    \item $\lit\substb$ is an equation, say $s\eql t\substb$, obtained from $\lit = (s\eql t)$ by replacing occurrences of $v$ in $t$ by $e$. In this case, by well-modedness of $\se_1,\lit,v \eql e,\se_2,\se_3$, variables of $s$ appear in $\se_1$ and $s$ does not contain $v$. Then the same property is maintained in $\se_1,\lit\substb,v\eql e, \se_2[\substb],\se_3[\substb]$, since $s$ remains in $\lit\substb$ and $\se_1$ does not change.
    \item $\lit\substb$ is an atom. Then replacing $v$ by $e$ in an output position of $\lit$, which gives $\lit\substb$, does not affect well-modedness.
  \end{enumerate}
  Hence, we got that $\se_1,\lit,v\eql e$ is well-moded. Now we can safely extend this sequence with $\se_2[\substb],\se_3[\substb]$, because variables in new occurrences of $e$ in $\se_2[\substb],\se_3[\substb]$ are preceded by $v\eql e$, and the relative order of the other variables does not change. Hence, the sequence $\se_1,\lit\substb,v\eql e, \se_2[\substb],\se_3[\substb]$ is well-moded.
 \item[\rm\emph{$\substb$ does not affect $\lit$.}] Then $\lit[\substb]=\lit$, the sequence $\se_1,\lit,v\eql e$ is well-moded and it can be safely extended with $\se_2[\substb],\se_3[\substb]$, obtaining the well-moded sequence $\se_1,\lit,v\eql\penalty1000 e,\allowbreak \se_2[\substb],\se_3[\substb]$.
\end{description}
Hence, we showed also in \emph{Case 2} that there exists a well-moded sequence of literals, namely, $\se_1,\lit[\substb],\allowbreak v\eql e, \se_2[\substb],\se_3[\substb]$, that corresponds to $\conjc_1 \land \conjc_2\substb \land v \eql e$. Hence, $\conjc_1 \land \conjc_2\substb \land v \eql e$ is well-moded.
\end{proof}

\sixLemWellModednessReduction*

\begin{proof}
Let $\Qr\myequiv \lit_1, \ldots, \lit_i,\ldots, \lit_n$, $\constr =\conjc_1\lor\cdots\lor\conjc_m$, and $\lbr \Qr \sep \constr \rbr $ be a well-moded state. We will use the notation $\Qrc$ for the conjunction of all literals in $\Qr$, i.e., $\Qrc = \lit_1\land  \cdots \land \lit_i \land \cdots \land \lit_n$. Assume that $\lit_i$ is the selected literal in reduction that gives $\lbr \Qr' \sep \constr' \rbr$ from $\lbr \Qr \sep \constr \rbr$. We consider four possible cases, according to the definition of operational semantics:

\emph{Case 1.} Let $\lit_i$ be a primitive constraint and $\constr' \not \myequiv \false$. Let $\constrd$ denote the DNF of $\constr\land \lit_i$. 

In order to prove that  $\lbr \Qr' \sep \constr' \rbr$ 
is well-moded, by the definition of $\sol$, it is sufficient to prove that $\lbr \Qr' \sep \step( \constrd) \rbr$ is well-moded. Since, obviously, $\lbr \Qr' \sep \constrd \rbr$ is a well-moded state, we have to show that state well-modedness is preserved by each rule of the solver. 

Since $\constr' \not \myequiv \false$, the step is not performed by any of the failure rules of the solver. For the rules {\sf M1--M8}, {\sf M11--M12}, {\sf D1}, and {\sf D2}, it is pretty easy to verify that $\lbr \Qr' \sep \step( \constrd) \rbr$ is well-moded. Therefore, we consider the other rules in more detail. We denote the disjunct of $\constrd$ on which the rule is applied by $\conjc_\constrd$. The cases below are distinguished by the rules:
\begin{description}
  \item[\rm {\sf Del}.] Here the same variable is removed from both sides of the selected equation. Assume $\sl_1,s\eql t,\sl_2$ is a well-moded sequence corresponding to $\Qrc'\land \conjc_\constrd$, and $s\eql t$ is the selected equation affected by one of the deletion rules. Well-modedness of $\sl_1,s\eql t,\sl_2$ requires that the variable deleted at this step from $s\eql t$ should occur in an output position in some other literal in $\sl_1$. Let $s'\eql t'$ be the equation obtained by the deletion step from $s\eql t$. Then $\sl_1,s'\eql t',\sl_2$ is again well-moded, which implies that $\Qrc'\land \step(\conjc_\constrd)$ is well-moded and, therefore, that $\lbr \Qr' \sep \step( \constrd) \rbr$ is well-moded.
 \item[\rm {\sf M9}.] Let $\Qrc'\land \conjc_\constrd$ be represented as $\Qrc'\land \Sv \inC f(\rh) \land \conjc'$, where $\Sv \inC f(\rh)$ is the membership atom affected by the rule. Note that then $\Qrc'\land \Sv \eql \Iv \land \Iv \inC f(\rh) \land \conjc'$ is also well-moded. Applying Lemma~\ref{lem:well:modedness:1}, we get that $\Qrc'\land \Sv \eql \Iv \land \Iv \inC f(\rh) \land \conjc'\substb$ is well-moded, where $\substb =\{\Sv \mapsto \Iv\}$. Then we get well-modedness of $\Qrc'\land \step(\conjc_\constrd)$, which implies well-modedness of $\lbr \Qr' \sep \step( \constrd) \rbr$.
  \item[\rm {\sf M10}.] Let $\Qrc'\land \conjc_\constrd$ be represented as $\Qrc'\land \Fv(\hdg) \inC f(\rh) \land \conjc'$, where $\Fv(\hdg) \inC f(\rh)$ is the membership atom affected by the rule. Note that then $\Qrc'\land \Fv(\hdg) \inC f(\rh) \land \Fv \eql f \land \conjc'$ is also well-moded. Applying Lemma~\ref{lem:well:modedness:1}, we get that $\Qrc'\land \Fv(\hdg)\substb \inC f(\rh) \land \Fv \eql f \land \conjc'\substb$ is well-moded, where $\substb=\{\Fv \mapsto f\}$. But it means that $\Qrc'\land \step(\conjc_\constrd)$ is well-moded, which implies that $\lbr \Qr' \sep \step( \constrd) \rbr$ is well-moded.
 \item[\rm {\sf E1}, {\sf E2}.] For these rules, well-modedness of $\Qrc'\land \step(\conjc_\constrd)$ is a direct consequence of Lemma~\ref{lem:well:modedness:1}.
 \item[\rm {\sf E3}.] Let $\Qrc'\land \conjc_\constrd$ be represented as $\Qrc' \land (\Sv,\hdg_1) \eqls \hdg_2 \land \conjc'$, where $(\Sv,\hdg_1) \eqls \hdg_2$ is the equation affected by the rule and $\Sv\not\in\var (\hdg_2)$. Then $\Qrc' \land \Sv \eql \hdg' \land \hdg_1 \eql \hdg''  \land \conjc'$ is also well-moded for some $\hdg'$ and $\hdg''$ with $(\hdg',\hdg'')= \hdg_2$. Applying Lemma~\ref{lem:well:modedness:1}, we get that $\Qrc' \land \Sv \eql \hdg' \land \hdg_1 \substb \eql \hdg''  \land \conjc'\substb $ is well-moded, where $\substb=\{\Sv \mapsto \hdg'\}$. Since $\hdg'$ and $\hdg''$ were arbitrary, it implies that $\Qrc'\land \step(\conjc_\constrd)$ and, therefore, $\lbr \Qr' \sep \step( \constrd) \rbr$ is well-moded.
 \item[\rm {\sf E4}.] Similar to the case of the rule {\sf E3}.
\end{description}
\emph{Case 2.} Let $\lit_i$ be a primitive constraint and $\constr'=\false$, where $\constr'= \sol(\constr \land \lit_i)$. Then by the operational semantics we have $\Qr'=\boxmy$ and the theorem trivially holds, since the state $\lbr \boxmy \sep \false \rbr$ is well-moded.

\emph{Case 3.} Let $\lit_i$ be an atom $p(t_1,\ldots,t_k,\ldots,t_l)$. Assume that $\Pr$ contains a clause of the form $p(r_1,\ldots,r_k,\ldots,r_l) \myrule \Bd$, where $\Bd$ denotes the body of the clause. Assume also that for the predicate $p$, the set $\{1,\ldots,k\}$ is the set of the input positions and $\{k+1,\ldots,l\}$ is the set of the output ones. Then we have 
\begin{alignat*}{1}
  \Qr = {} & \lit_1, \ldots, \lit_{i-1}, p(t_1,\ldots,t_k,\ldots,t_l),\lit_{i+1},\ldots, \lit_n,\\
  \Qr'= {} & \lit_1, \ldots, \lit_{i-1}, t_1 \eql r_1,\ldots, t_k\eql r_k, \ldots, t_l\eql r_l, \Bd, \lit_{i+1},\ldots, \lit_n,\\
  \constr' = {} & \constr = \conjc_1\lor \cdots \lor \conjc_m.
\end{alignat*}

From well-modedness of the state $\lbr \Qr\sep \constr \rbr$ we know that for all $1\le j\le m$, the literals from $\lit_1, \ldots, \lit_{i-1}, \lit_{i+1},\ldots, \lit_n$ and $\conjc_j $ can be reordered in two sequences of literals $\sl^1_j$ and $\sl^2_j$ in such a way that the sequence $\sl^1_j,p(t_1,\ldots,t_k,\ldots,t_l),\sl^2_j$ is well-moded. Then we have $\var(t_1,\ldots,$ $t_k)\subseteq\outvar (\sl^1_j)$. Therefore, we obtain that the sequence
\begin{equation}\label{wm:1}
  \sl^1_j, t_1 \eql r_1,\ldots, t_k\eql \penalty1000r_k ,\sl^2_j
\end{equation}
is well-moded for all $1\le j\le m$.

From well-modedness of $p(r_1,\ldots,r_k,\ldots,r_l) \myrule \Bd$ we know that $\var(r_{k+1},\ldots,r_l)\subseteq \outvar(\Bd) \cup \var(r_1,\ldots,r_k)$. By item \ref{wel} of the definition of program well-modedness,  the literals of $\Bd$ can be put into a well-moded sequence, written, say, as $B_1,\ldots,B_q$, such that for each $1\le u\le q$ and $v\in \invar(B_u)$ we have $v\in \outvar(B_{u'})$ for some $u'<u$, or $v \in \var(r_1,\ldots,r_k)$. From then we can say that the sequence 
\begin{equation}\label{wm:2}
t_1 \eql r_1,\ldots, t_k\eql \penalty1000r_k,\allowbreak B_1,\ldots,B_q, t_{k+1} \eql r_{k+1},\ldots, t_l\eql r_l
\end{equation}
is well-moded.

From (\ref{wm:1}) and (\ref{wm:2}), by the definition of well-modedness, we can conclude that
\begin{equation}\label{wm:3}
 \sl^1_j,t_1 \eql r_1,\ldots, t_k\eql \penalty1000r_k,\allowbreak B_1,\ldots,B_q, t_{k+1} \eql r_{k+1},\ldots, t_l\eql r_l,\sl^2_j
\end{equation}
is well-moded for all $1\le j\le m$. By construction, the literals in (\ref{wm:3}) are exactly those from $\Qrc'\land \conjc_j$ for $1\le j\le m$. It means that $\lbr \Qr' \sep \conjc_j  \rbr$ is well-moded for all $1\le j\le m$, which implies that $\lbr \Qr' \sep \constr'  \rbr$ is well-moded.

\emph{Case 4.} If $\defn_{\prog}(\lit_i)=\emptyset$, then  $\Qr'=\boxmy$, $\constr'=\false$, and the theorem trivially holds.
\end{proof}

\sevenCorWellModedness*
\begin{proof}
By the definition of well-modedness, since $\constr$ is well-moded, the state $\lbr a\eql a\sep \constr \rbr$ is also well-moded, where $a$ is an arbitrary function symbol. By the operational semantics, we have the reduction $\lbr a\eql a \sep \constr \rbr \rightarrowtail \lbr \boxmy \sep \sol(a\eql a \land \constr) \rbr $. By Lemma \ref{lem:well:modedness:2}, we get that $\lbr \boxmy \sep \sol(a\eql a \land \constr) \rbr $ is also well-moded and, hence, $\sol(a\eql a \land \constr)$ is well-moded. By the definition of $\sol$ and the rules of the solver, it is straightforward to see that $\sol(a\eql a \land \constr) = \sol(\constr)$. Hence, $ \sol(\constr)$ is well-moded.
\end{proof}

\eightThmSolvedForm*

\begin{proof}
By the Corollary \ref{cor:wellmoded}, the constraint  $\constr'$  is  well-moded.  If $\constr'$ is $\true$ then it  is already solved. Consider the case when $\constr'$ is not  $\false$. Let $\constr'=\conjc_1\lor\cdots\lor\conjc_m$. Since $\constr'\neq \false$, by the Theorem \ref{thm:equivalence}  $\constr'$ is partially solved. It means that  each $\conjc_j$, $1\le j\le m$, is partially solved and well-moded. By definition, $\conjc_j$ is well-moded if there exists a permutation of its literals $\pmc_1, \ldots, \pmc_i,\ldots, \pmc_n$ which satisfies the well-modedness property. Assume $\pmc_1, \ldots , \pmc_{i-1}$ are solved. By this assumption and the definition of well-modedness, each of $\pmc_1, \ldots , \pmc_{i-1}$ is an equation whose one side is a variable that occurs neither in its other side nor in any other primitive constraint. Then well-modedness of $\conjc_j$ guarantees that the other sides of these equations are ground terms. Assume by contradiction that $\pmc_i $ is partially solved, but not solved. If $\pmc_i$ is a membership constraint, well-modedness of $\conjc_j$ implies that $\pmc_i$ does not contain variables and, therefore, can not be partially solved. Now let $\pmc_i$ be an equation. Since all variables in $\pmc_1, \ldots, \pmc_{i-1}$ are solved, they can not appear in $\pmc_i$. From this fact and well-modedness of $\conjc_j$, $\pmc_i$ should have at least one ground side. But then it can not be partially solved. The obtained contradiction shows that $\constr'$ is solved. 
\end{proof}

\nineThmMainWellModed*

\begin{proof}
We prove a slightly more general statement: Let $ \lbr \Qr \sep \true \rbr \rightarrowtail \cdots \rightarrowtail   \lbr \Qr' \sep \constr'\rbr$ be a derivation with respect to a well-moded program, starting from a well-moded goal $\Qr $ and ending with $\Qr'$ that is either $\boxmy$ or consists only of atomic formulas without arguments (propositional constants). If $\constr'\neq \false$, then $\constr'$ is solved.

To prove this statement, we use induction on the length $n$ of the derivation. When $n=0$, then $\constr'=\true$ and it is solved. Assume the statement holds when the derivation length is $n$, and prove it for the derivation with the length $n+1$. Let such a derivation be $ \lbr \Qr \sep \true \rbr \rightarrowtail \cdots \rightarrowtail   \lbr \Qr_n \sep \constr_n\rbr \rightarrowtail   \lbr \Qr_{n+1} \sep \constr_{n+1}\rbr $. Assume that $\Qr_{n+1}$ that is either $\boxmy$ or consists only of propositional constants. According to the operational semantics, there are two possibilities how the last step is made:
\begin{enumerate}
  \item $\Qr_n$ has a form (modulo permutation) $\lit,p_1,\ldots,p_m$, $m\ge 0$, where $\lit$ is primitive constraint, the $p$'s are propositional constants, $\Qr_{n+1}=p_1,\ldots,p_m$, and $\constr_{n+1}=\sol(\constr_n\land \lit)$.
  \item $\Qr_n$ has a form (modulo permutation) $q,p_1,\ldots,p_m$, $m\ge 0$, where $q$ and $p$'s are propositional constants, the program contains a clause $q \leftarrow q_1,\ldots,q_k$, $k\ge 0$, where all $q_i$, $1\le i\le k$, are propositional constants, $\Qr_{n+1}=q_1,\ldots,q_k,p_1,\ldots,p_m$, and $\constr_{n+1}=\constr_n$.
\end{enumerate}
In the first case, by the $n$-fold application of Lemma~\ref{lem:well:modedness:2} we get that $\lbr \Qr_n \sep \constr_n\rbr$ is well-moded. Since the $p$'s have no influence on well-modedness (they are just propositional constants), $\constr_n\land \lit$ is well-moded and hence it is solvable. By Theorem~\ref{lem:solvedform} we get that if $\constr_{n+1}=\sol(\constr_n\land \lit)\neq \false$, then $\constr_{n+1}$ is solved.

In the second case, since $G_n$ consists of propositional constants only, by the induction hypothesis we have that if $\constr_n$ is not $\false$, then it is solved. But $\constr_n=\constr_{n+1}$.  It finishes the proof.
\end{proof}

\tenLemKIFPartiallySolved*

\begin{proof}
 Let $\conjc$ be a partially solved conjunction of primitive constraints. Then, by the definition, each primitive constraint $\pmc$ from $\conjc$ should be either solved in $\conjc$, or should have one of the following forms:
\allowdisplaybreaks
\begin{itemize}
\item Membership atom:
  \begin{itemize}
   \item $\fu(\hdg_1,\Sv,\hdg_2) \inC \fu(\rh)$.
    \item $(\Sv,\hdg) \inC \rh $  where  $\hdg \neq \epsilon $ and $\rh$ has the form  $\rh_1 \concs \rh_2 $ or $\rh_1^\strs$.
\end{itemize}
\item Equation: 
   \begin{itemize}
     \item $ (\Sv,\hdg_1)\eql (\Svy,\hdg_2) $ where $\Sv \not = \Svy$, $\hdg_1\neq \epsilon$ and $\hdg_2\neq \epsilon$.
     \item $ (\Sv,\hdg_1)\eql (\tseq,\Svy,\hdg_2) $, where $\Sv\not \in \var(\tseq)$, $\hdg_1\neq \epsilon$, and $\tseq\neq \epsilon$.  The variables $\Sv $ and $\Svy$ are not necessarily distinct.
  \item $\fu(\hdg_1,\Sv,\hdg_2)\eql \fu(\hdg_3,\Svy,\hdg_4)$ where $(\hdg_1,\Sv,\hdg_2)$ and $(\hdg_3,\Svy,\hdg_4)$ are disjoint.
\end{itemize}
 \end{itemize}
However, $\pmc$ is also a KIF constraint. By the definition of KIF form, none of the above mentioned forms for membership atoms and equations are permitted. Hence, $\pmc$ is solved in $\conjc$ and, therefore, $\conjc$ is solved. It implies the lemma.
\end{proof}

\elevenThmKIFSolved*

\begin{proof}
 By Theorem~\ref{thm:equivalence}, $\constr'$ should be in a partially solved form. It is also in the KIF form, as we noted above. Then, by Lemma~\ref{lem:kif:partially:solved}, $\constr'$ is solved.
\end{proof}

\twelveThmKIFDerivation*
\begin{proof}
  Since the reduction preserves KIF states, $\constr'$ is in the KIF form. Since the derivation is finished and $\constr'\neq \false$, by the definition of finished derivation, $\constr'$ is partially solved. By Lemma~\ref{lem:kif:partially:solved}, we conclude that $\constr'$ is solved.
\end{proof}

\end{document}